\documentclass[10pt,a4paper,onecolumn]{IEEEtran}
\usepackage{blindtext, graphicx}
\usepackage{tabularx}
\usepackage{pgf,tikz,pgfplots}
\usetikzlibrary[patterns]
\usepackage{algorithm}
\usepackage{algpseudocode}
\usepackage{hyperref}
\pgfplotsset{compat=1.14}
\usepackage{mathrsfs}
\usepackage{authblk}
\usetikzlibrary{arrows}
\hypersetup{
 colorlinks,
 linkcolor={blue!100!black},
 citecolor={blue!100!black},
 urlcolor={blue!80!black}
}

\usepackage{amsmath,amssymb,amsfonts,amsthm,graphicx,bm,bbm,dsfont}
\DeclareMathOperator{\decimal}{decimal}
\newtheorem{theorem}{Theorem}
\newtheorem{remark}{Remark}

\newtheorem{example}{Example}
\newtheorem{lemma}{Lemma}

\usepackage{soul}

\DeclareSymbolFont{bbold}{U}{bbold}{m}{n}
\DeclareSymbolFontAlphabet{\mathbbold}{bbold}
\newcommand{\1}{\mathbbold{1}}

\usepackage{color}

\newcommand{\cB}{\mathcal{B}}
\newcommand{\cC}{\mathcal{C}}
\newcommand{\cD}{\mathcal{D}}

\newcommand{\cI}{\mathcal{I}}

\newcommand{\cS}{\mathcal{S}}

\newcommand{\bolda}{\mathbf{a}}

\newcommand{\boldc}{\mathbf{c}}
\newcommand{\boldd}{\mathbf{d}}

\newcommand{\boldm}{\mathbf{m}}

\newcommand{\bolds}{\mathbf{s}}

\newcommand{\boldx}{\mathbf{x}}
\newcommand{\boldy}{\mathbf{y}}

\makeatletter
\def\namedlabel#1#2{\begingroup
	\def\@currentlabel{#2}%
	\label{#1}\endgroup
}
\makeatother
\begin{document}
\title{Robust Indexing for the Sliced Channel: Almost Optimal Codes for Substitutions and Deletions}
\author[1]{\textbf{Jin Sima}}
\author[2]{\textbf{Netanel Raviv}}
\author[3]{\textbf{Jehoshua Bruck}}
\affil[1]{\normalsize{Department of Electrical and Computer Engineering, University of Illinois Urbana Champaign
}}
\affil[2]{\normalsize{Department of Computer Science and Engineering, Washington University in Saint Louis
}}
\affil[3]{\normalsize{Department of Electrical Engineering, California Institute of Technology
}}
\maketitle

\begin{abstract}
	 Encoding data as a set of unordered strings is receiving great attention as it captures one of the basic features of DNA storage systems. However, the challenge of constructing optimal redundancy codes for this channel remained elusive. In this paper, we address this problem and present an order-wise optimal construction of codes that are capable of correcting multiple substitution, deletion, and insertion errors for this channel model. The key ingredient in the code construction is a technique we call \textit{robust indexing}: simultaneously assigning indices to unordered strings (hence, creating order) and also embedding information in these indices. 
  The encoded indices are resilient to substitution, deletion, and insertion errors, and therefore, so is the entire code.
\end{abstract}

\thispagestyle{empty}



\section{Introduction}\label{section:introduction}
The interest in storing data in synthetic DNA is drastically increasing lately, due to its advantages of ultra high data density and longevity over other storage media. Tremendous progress has been made in synthesizing and sequencing technologies, which brings about a new era in large-scale DNA storage.  Prototype implementations of DNA storage stored 643KB data in~\cite{Church} and 739KB data in~\cite{Goldman} respectively, followed by many experiments~\cite{antkowiak2020low,blawat16,bornholt2016dna,chandak2019improved,Erlich,Microsoft,Yazdi,Yazdi17} that improved the storage data size. The largest data size achieved in DNA storage is  200MB~\cite{Microsoft}. 

In DNA storage systems, data is represented by strings of four nucleotides that make up the synthesized DNA molecules.  
One of the key features that distinguishes DNA storage from conventional storage media is that data is encoded as an unordered set of short strings, rather than a single long string. This is because current technology cannot synthesize a single DNA string long enough to encode the entire data. The typical length of a short DNA string is several hundreds.

When writing the data, these short strings are synthesized into DNA molecules and stored in a DNA pool. When reading the data, a Polymerase Chain Reaction (PCR) process is used for retrieving the targeted parts of the data. In the PCR process, the number of copies of each targeted DNA molecule is significantly amplified.  Then,
the pool of amplified DNA molecules is sampled and sequenced, producing multiple reads of the short strings that encode the data. In the above reading and writing process, sequencing and synthesis errors can occur, resulting in substitution, deletion and insertion errors in the DNA strings. One way to correct these errors is to cluster the erroneous reads by similarity and use a sequence reconstruction algorithm \cite{Batu} on each cluster to recover the original strings. Yet, such clustering and reconstruction algorithms at the decoder cannot correct writing (synthesis) errors, because the sampled and sequenced pool of DNA molecules are amplified versions of the synthesized ones, which contain writing errors.  
Thus, error-correcting codes for DNA storage come into play.  

While many coding theoretic results have been obtained for various channel models concerning different physical aspects of DNA storage~\cite{Chang, Gabrys1, Kiah, Raviv}, this paper focuses on coding over unordered sets, which captures some basic features in the writing and reading processes described above. Specifically, consider encoding data into~$M$ strings of length~$L$. The decoder wishes to recover the data from erroneous versions of the~$M$ strings, which contain substitution, deletion and insertion errors. This model has been extensively investigated recently. 
The works of~\cite{CodingOverSets,Song,wei2021improved} proposed constructions and upper bounds for coding over an unordered set of strings with sequence loss and symbol substitutions. 
To deal with unordered strings, one of the natural approaches is to assign~$\log M$ bits to each string for indexing so that the strings are ordered. Such index-based construction was considered in~\cite{AnchorbasedCor} and~\cite{ClusterCodes}, which proposed code constructions that correct errors in the indices. It was proved in \cite{Shuffling}, from an information theoretic view, that the channel capacity for communicating over an unordered set of $M$ binary symmetric channels can be achieved by using index-based schemes, under some channel parameter constraints. The capacity results in~\cite{Shuffling} were later extended to more general settings \cite{Achievingcapacity,weinberger2022dna}. 
While index-based schemes are optimal in terms of coding rate, they are sub-optimal in terms of coding redundancy for correcting a small number of errors.   
The work of~\cite{sliced} showed that for a constant number~$K$ of substitution errors, the optimal redundancy for coding over $M$ strings of length $L$ was~$O(K\log ML)$ bits. This is less than the $O(M)$ bits of redundancy required in  index-based schemes \cite{CodingOverSets} for $M$ larger than $L$, which is the common case in DNA storage. 

Though an explicit code with~$O(K^2\log ML)$ bits of redundancy was given in \cite{sliced} for correcting $K$ substitutions over an unordered set of $M$ strings of length $L$, no order-wise optimal and explicit code construction was known. 
In this paper, we close this gap and propose order-wise optimal code constructions that achieve~$O(K\log ML)$ redundancy for~$K$ substitution errors, based on a technique we call \emph{robust indexing}. It is assumed throughout the paper that $M\ge 2$, since the case $M=1$ reduces to ordinary channel coding. Our first main result is as follows
\begin{theorem}\label{theorem:main}
For integers~$M,L$,~$K$, and~$L'\triangleq 3\log M  + 4K^2+2$. If~$L'+4KL'+2K\log (4KL')\le L$, then there exists an explicit $K$-substitution code, computable in~$poly(M,L,K)$ time, that has redundancy $ 2K\log ML + (12K+2)\log M+O(K^3)+O(K\log\log ML)$.
\end{theorem}
Roughly speaking, the proof of Theorem \ref{theorem:main} relies on the following---Instead of directly assigning an index to each short string, as is done in index based coding schemes,  we \emph{embed information into indices}. Note that to combat errors, the indexing bits themselves constitute an error correcting codebook  and information is carried through choices of the codebook. The idea of encoding information through codebook choices also appeared in \cite{AnchorbasedCor}. The difference between our construction and the one in \cite{AnchorbasedCor} is that the construction in \cite{AnchorbasedCor} is index-based, which inherently requires redundancy which is at least linear in $M$, while our construction uses the data itself for indexing. In addition, 
our \emph{robust indexing} algorithm generates indexing bits in a greedy manner and is computationally efficient. Our algorithm also applies to deletion/insertion errors by considering the deletion/insertion distance metric, instead of the Hamming distance metric. By  using~$K$-deletion correcting  codes~\cite{kdeletion} for a single string, we propose a code that corrects~$K$ deletions with~$O(K\log ML)$ redundancy, which is our second main result.
\begin{theorem}
For integers~$M,L,K$, and~$L'\triangleq 3\log M  + 4K^2+2$. If~$L'+4KL'+2K\log (4KL')\le L$, then there exists a $K$-deletion code, encodable and decodable  in $O(N^{2K+1})$ time, that has redundancy $ 4K\log ML + (12K+2)\log M+O(K^3)+o(\log ML)$. 
\end{theorem}
\begin{remark}
The encoding/decoding complexity $O(N^{2K+1})$     comes from the complexity of the codes in \cite{kdeletion}. The robust indexing algorithm, which reduces the problem of coding over unordered strings to that of coding over a single string, has complexity $poly(M,L,K)$.
\end{remark}
Other related problem settings include: permutation channels~\cite{Kovacevic,perm2,makur2020coding,tang2023capacity,perm1}, which consider string errors only, and torn paper coding \cite{bar2022adversarial,shomorony2021torn}. See \cite{shomorony2022information} for a broader survey of the related problems.

The rest of the paper is organized as follows. Section~\ref{section:preliminary} presents the notations and the channel model. In Section~\ref{section:robusti} we provide an order-wise optimal code construction for substitution errors.  The~\emph{robust indexing} algorithm is given in Section~\ref{section:computingfh}. In Section~\ref{section:computingfd} we apply \emph{robust indexing} to deletion errors and propose a deletion correcting code construction.

\section{Preliminaries}\label{section:preliminary}
We focus on the binary alphabet~$\{0,1\}$. For a set~$S$ and an integer~$m$,
denote by~$\binom{S}{m}$ the family of all subsets of~$m$ different elements in~$S$, and by~$\binom{S}{\le m}=\bigcup^{M}_{i=1}\binom{S}{m}$ the family of all subsets of~$S$ with no more than~$m$ elements. 
For an integer~$\ell$, let~$\{0,1\}^{\le \ell}$ be the set of all binary strings with length at most~$\ell$.
 In our channel model, it is assumed that the data is given as a binary string and encoded in an unordered set of~$M$ different strings~$\{\boldx_i\}^M_{i=1}$ of length~$L$. In this paper, a set $\{\boldx_i\}^M_{i=1}\in\binom{\{0,1\}^L}{M}$  is referred to as a \textit{word}, and each element~$\boldx_i$ in a word is referred to as a \textit{string}. Note that in our settings,  a code is a set of words, and a \textit{codeword} is a word in the code, rather than a string as in classic coding theoretic settings.

 The assumption that the strings~$\boldx_i$,~$i\in[M]$, in a word are different stems from the fact that the sampling and  sequencing procedures in the reading process cannot detect repeated strings in the word~$\{\boldx_i\}^M_{i=1}$.  Moreover, it follows from the definition of code redundancy that will be presented later, that the asymptotic redundancy of a code is not affected by allowing repeated strings in the codeword, when~$M=o(2^L)$. 
 
In the considered channel, a word~$\{\boldx_i\}^M_{i=1}$ is subject to substitution, deletion and insertion errors . In this paper, we propose codes for correcting substitution errors and deletion errors separately. The presented deletion code is capable of correcting deletion/insertion errors as well. A~$K$-substitution error is an operation that flips at most~$K$ bits in the    word. Each bit flip can occur in any of the strings~$\boldx_i$,~$i\in[M]$, where~$[M]\triangleq\{1,\ldots,M\}$. For any word~$\{\boldx_i\}^M_{i=1}\in\binom{\{0,1\}^L}{M}$, define its Hamming ball~$\cB^H_K(\{\boldx_i\}^M_{i=1})\in\binom{\{0,1\}^L}{\le M}$ as the set of all possible outcomes of a~$K$-substitution error in~$\{\boldx_i\}^M_{i=1}$. Note that a word can have less than $M$ strings after a $K$-substitution error, if two strings in the word become identical after substitution errors.  A~$K$-substitution code~$\cC^H$ is an ensemble of words~$\{\boldx_i\}^M_{i=1}\in\binom{\{0,1\}^L}{M}$ such that for any two words~$S_1,S_2\in\cC^H$, we have that~$\cB^H_K(S_1)\cap\cB^H_K(S_2)=\emptyset$.
Similarly, a~$K$-deletion error is an operation that deletes at most~$K$ bits in a word. Each deletion can occur at any of the strings $\boldx_i$, $i\in[M]$. 
For a set $\{\boldx_i\}^M_{i=1}\in\binom{\{0,1\}^L}{M}$, its deletion ball~$\cB^D_K(\{\boldx_i\}^M_{i=1})\subseteq\binom{\{0,1\}^{\le L}}{\le M}$ is the collection of outcomes of a~$K$-deletion error in~$\{\boldx_i\}^M_{i=1}$. A~$K$-deletion code~$\cC^D$ is a collection  of words~$\{\boldx_i\}^M_{i=1}\in\binom{\{0,1\}^L}{M}$ such that the deletion balls of any two words~$S_1,S_2\in\cC^D$ do not intersect. The redundancy of a $K$-substitution or $K$-deletion code~$\cC$ is defined as~$r(\cC)=\log \binom{2^L}{M}-\log |\cC|$.
\begin{example}
    Let $\{001,101\}$ be a word in $\binom{\{0,1\}^3}{2}$ and let $K=1$. Then its Hamming ball $\cB^H_K(\{001,101\})$ is given by
    \begin{align*}
        \cB^H_K(\{001,101\})=\Big\{\{001,101\},\{101\},\{011,101\},\{000,101\},\{001\},\{001,111\},\{001,100\}\Big\}.
    \end{align*}
    Its deletion ball $\cB^D_K(\{001,101\})$ is given by 
    \begin{align*}
        \cB^D_K(\{001,101\})=\Big\{\{001,101\},\{01,101\},\{00,101\},\{001,01\},\{001,11\},\{001,10\}\Big\}.
    \end{align*}    
\end{example}
Our code constructions make use of Reed-Solomon (RS) code~\cite{Ronny}. To this end, let $RS_k:\{0,1\}^n\to\{0,1\}^{2k\log n+o(\log n)}$ be a function which returns the redundancy of the input word, when encoded using a systematic Reed-Solomon code of minimum distance $2k+1$, i.e., such that $(\boldm,RS_k(\boldm))$ is a codeword in an RS code of minimum distance $2k+1$ for any $\boldm\in\{0,1\}^n$. In detail, the function chooses the smallest $q$ such that $q\ge \frac{n}{\lceil\log_2 q\rceil}-1$, splits the input into blocks of length $\lceil\log_2 q\rceil$, considers each block as an element in $\mathbb{F}_q$,  applies standard systematic RS encoding, and outputs only the redundancy bits. 

In addition, combinatorial numbering maps~\cite{Knuth} are used in the robust indexing algorithm. Specifically, for integers~$m$ and~$n$, there exist a  map~$F_{com}:[\binom{n}{m}]\rightarrow \binom{[n]}{m}$ that maps an integer~$d\in[\binom{n}{m}]$ to a set of~$m$ different elements in~$[n]$, and a map~$F_{perm}:[n!]\rightarrow S_n$ that maps an integer~$d\in[n!]$ into a permutation on~$n$ elements. 
Both $F_{com}$ and $F_{perm}$ can be computed in $poly(m,n)$ time.

\section{Robust indexing for substitution errors}\label{section:robusti}
In this section we describe our code constructions, based on the idea of robust indexing. 
These codes correct substitution errors, and have redundancy~$O(K\log ML)$, which is order-wise optimal whenever~$K$ is at most $O(\min\{L^{1/3},L/\log M\})$. The code presented in this section proves Theorem \ref{theorem:main}.

Since the codewords consist of unordered strings, we assign indexing bits to each string such that the strings in each codeword are ordered. However, instead of directly assigning the indices~$1,\ldots,M$ to each string using $\log (M)$ bits, we embed information into the indexing bits. In other words, we use the information bits themselves for the purpose of indexing. 


Specifically, for a codeword~$W=\{\boldx_i\}^M_{i=1}\in\binom{\{0,1\}^L}{M}$,
we choose the first~$L'$ bits~$(x_{i,1},x_{i,2},\ldots,x_{i,L'})\in\{0,1\}^{L'}$,~$i\in[M]$ in each string~$\boldx_i$ as the indexing bits, for some $L'$ that is defined shortly, and encode information in them. Then, the strings in~$\{\boldx_i\}^M_{i=1}$ are sorted according to the lexicographic order~$\pi$ of the indexing bits $(x_{i,1},x_{i,2},\ldots,x_{i,L'})$,~$i\in[M]$, where\footnote{For two binary strings $\boldx=(x_1,\ldots,x_m)$ and $\boldy=(y_1,\ldots,y_m)$,  we say that $\boldx>\boldy$ if there exists an index $i$ such that $x_j=y_j$ for $j\in[i-1]$ and $x_i>y_i$.}~$(x_{\pi (i),1},x_{\pi(i),2},\ldots,x_{\pi(i),L'})<(x_{\pi (j),1},x_{\pi(j),2},\ldots,x_{\pi(j),L'})$ for~$i<j$. 
Once the strings in~$\{\boldx_i\}^M_{i=1}$ are ordered, it suffices to use a Reed-Solomon code to protect the concatenated string~$(\boldx_{\pi(1)},\ldots,\boldx_{\pi(M)})$, and thus the codeword~$\{\boldx_i\}^M_{i=1}$, from~$K$ substitution errors.

One key issue with this approach is that the indexing bits and their lexicographic order can be disrupted by substitution errors 
, in cases when: (a) Two erroneous strings $(x'_{i,1},x'_{i,2},\ldots,x'_{i,L'})$ and $(x'_{j,1},x'_{j,2},\ldots,x'_{j,L'})$ of the strings $(x_{i,1},x_{i,2},\ldots,x_{i,L'})$ and $(x_{j,1},x_{j,2},\ldots,x_{j,L'})$, respectively, might be identical, i.e., $(x'_{i,1},x'_{i,2},\ldots,x'_{i,L'})=(x'_{j,1},x'_{j,2},\ldots,x'_{j,L'})$, or alternatively might switch values, i.e., $(x'_{i,1},x'_{i,2},\ldots,x'_{i,L'})=(x_{j,1},x_{j,2},\ldots,x_{j,L'})$ and $(x'_{j,1},x'_{j,2},\ldots,x'_{j,L'})=(x_{i,1},x_{i,2},\ldots,x_{i,L'})$, both cases preclude RS decoding. (b) The string $(x_{j,1},x_{j,2},\ldots,x_{j,L'})$ may have a different value as a result of substitution errors, causing a change in its lexicographic order. Further, the decoder might not know what the correct decoding is, as there are multiple ways of correcting the erroneous strings which result in a set of M indices whose pairwise minimum distance is at least $2K+1$.

To deal with such cases, we present
a technique we call \emph{robust indexing}, which protects the indexing bits from substitution errors.  
The basic ideas of~\emph{robust indexing} are as follows: 

(1) Constructing the indexing bits~$\{(x_{i,1},x_{i,2},\ldots,x_{i,L'})\}^M_{i=1}$ such that the Hamming distance between any two distinct $(x_{i,1},x_{i,2},\ldots,x_{i,L'})$ and~$(x_{j,1},x_{j,2},\ldots,x_{j,L'})$ is at least ~$2K+1$, i.e., the strings~$\{(x_{i,1},x_{i,2},\ldots,x_{i,L'})\}^M_{i=1}$ constitute an error correcting code in classical settings where each codeword is a string. Note that by adding the distance $2K+1$ constraint on strings~$\{(x_{i,1},x_{i,2},\ldots,x_{i,L'})\}^M_{i=1}$, we avoid case (a) mentioned above. Moreover, we can  compute a mapping between the strings in~$\{(x_{i,1},x_{i,2},\ldots,x_{i,L'})\}^M_{i=1}$ and the strings in the erroneous version~$\{(x'_{i,1},x'_{i,2},\ldots,x'_{i,L'})\}^M_{i=1}$ by using a minimum Hamming distance criterion. The mapping correctly identifies a string $(x_{i,1},x_{i,2},\ldots,x_{i,L'})$ from its erroneous version~$(x'_{i,1},x'_{i,2},\ldots,x'_{i,L'})$ for $i\in[M]$.

(2) Using additional redundancy to protect the set of indexing bits~$\{(x_{i,1},x_{i,2},\ldots,x_{i,L'})\}^M_{i=1}$ from substitution errors. More specifically, we protect the indicator vector, which is a vector representation of the set $\{(x_{i,1},x_{i,2},\ldots,x_{i,L'})\}^M_{i=1}$ and is defined later, by using Reed-Solomon code. 
Since the string $(x_{i,1},x_{i,2},\ldots,x_{i,L'})$ can be identified from its erroneous version $(x'_{i,1},x'_{i,2},\ldots,x'_{i,L'})$ as described in (1), 
we avoid case (b) mentioned above after recovering the values of the strings  $\{(x_{i,1},x_{i,2},\ldots,x_{i,L'})\}^M_{i=1}$.

Note that in the robust indexing technique, we encode part of the data in the code~$\{(x_{i,1},x_{i,2},\ldots,x_{i,L'})\}^M_{i=1}$ through different choices of the code. Namely, in the first $L'$ bits, the space of all possible words is the set of all (ordinary) codes of length $L'$ and minimum distance at least $2K+1$, and encoding is an injective map from part of the information bits to this space. 


In the following, we present more details about the ideas in (1) and then in (2). We start with (1) and give a definition of the space of codes of length $L'$ and minimum distance at least $2K+1$ that we are interested in. The redundancy introduced by mapping information into this space and the algorithm for computing the mapping are discussed in the sequel. 
For an integer~$\ell$, let~$\1_\ell$ be the all~$1$'s vector of length~$\ell$.
Define~$\mathcal{S}^H$ as the set of all length~$L'$ codes that have cardinality~$M$ and minimum Hamming distance at least~$2K+1$ and  contain~$\1_{L'}$ (including the all 1's vector is a technical requirement, that will be made clear in the sequel.), that is,
\begin{align*}
    \cS^H\triangleq\left\{\{\bolda_1,\ldots,\bolda_M\}:
    \bolda_1= \1_{L'},\mbox{ $\bolda_i\in\{0,1\}^{L'}$, and }
d_H(\bolda_i,\bolda_j)\ge 2K+1\mbox{ for every distinct }i,j\in[M]\right\}.
\end{align*}
The following lemma gives a lower bound on the size of~$\cS^H$ and is obtained using a counting argument.
\begin{lemma}\label{lemma:lowerboundS}
	Let $Q=\sum^{2K}_{i=0}\binom{L'}{i}$ be the size of a Hamming ball of radius~$2K$ in~$\{0,1\}^{L'}$. We have that
	\begin{align}\label{equation:numberofcodes}
	|\mathcal{S}^H|\ge \frac{(2^{L'}-MQ)^{M-1}}{(M-1)!},
	\end{align}
\end{lemma}
\begin{IEEEproof}
Define the set of \textit{ordered} tuples $$\cS^H_T=\left\{(\bolda_1,\ldots,\bolda_{M}):\bolda_1= \1_{L'},\mbox{ $\bolda_i\in\{0,1\}^{L'}$, and }d_H(\bolda_i,\bolda_j)\ge 2K+1 \mbox{ for distinct }i,j\in[M]\right\}$$ such that for each tuple~$(\bolda_1,\ldots,\bolda_{M})\in\cS^H_T$, we have that~$\{\bolda_1,\ldots,\bolda_{M}\}\in\cS^H$.
We show that~$|\cS^H_T|\ge \prod^{M}_{i=2}[2^{L'}-(i-1)Q]$, by finding~$\prod^{M}_{i=2}[2^{L'}-(i-1)Q]$ distinct tuples in~$\cS^H_T$. Let~$\bolda_1=\1_{L'}$. We select~$\bolda_2,\ldots,\bolda_{M}$ sequentially such that each selected string~$\bolda_i$,~$i\in[2,M]\triangleq\{2,\ldots,M\}$, is of Hamming distance at least~$2K+1$ from each one of~$\bolda_{1},\ldots,\bolda_{i-1}$. The tuple~$(\bolda_1,\ldots,\bolda_M)$ selected in this way has pairwise Hamming distance at least~$2K+1$ and thus belongs to~$\cS^H_T$. 

Since the number of strings having Hamming distance at most~$2K$ from at least one of~$\bolda_1,\ldots,\bolda_{i-1}$ is at most~$(i-1)Q$, there are at least~$2^{L'}-(i-1)Q$ possible choices of~$\bolda_i$ that have Hamming distance at least~$2K+1$ from each one of~$\bolda_{1},\ldots,\bolda_{i-1}$. Therefore, the total number of ways for selecting tuples~$(\bolda_1,\ldots,\bolda_{M})$ is at least~$\prod^{M}_{i=2}[2^{L'}-(i-1)Q]$.
Since in the above selection of tuples~$(\bolda_1=\1_{L'},\ldots,\bolda_M)\in\cS^H_T$, there are~$(M-1)!$ tuples that correspond to the same set~$\{\bolda_1=\1_{L'},\bolda_2,\ldots,\bolda_{M}\}$ in~$\mathcal{S}^H$, we have that
\begin{align*}
	|\cS^H|=|\cS^H_T|/(M-1)!\ge \prod^{M}_{i=2}[2^{L'}-(i-1)Q]/(M-1)!\ge \frac{(2^{L'}-MQ)^{M-1}}{(M-1)!}.\qedhere
\end{align*}
\end{IEEEproof}
According to~\eqref{equation:numberofcodes}, there exists an invertible mapping~$F^H_S:[\lceil\frac{(2^{L'}-MQ)^{M-1}}{(M-1)!}\rceil]\rightarrow \binom{\{0,1\}^{L'}}{M}$, computed in~$O(2^{ML'})$ time using brute force, that maps an integer~$d\in\left[\lceil\frac{(2^{L'}-MQ)^{M-1}}{(M-1)!}\rceil\right]$ 
to a code~$F^H_S(d)\in \mathcal{S}^H$. In the next section, we will present a polynomial time algorithm that computes a map~$F^H_S(d)$ for any $d\in\left[\binom{2^{L'}-(M-1)Q+M-1}{M-1}\right]$ such that $F^H_S(d)\in \cS^H$ and $F^H_S(d_1)\ne F^H_S(d_2)$ for any distinct $d_1,d_2\in \left[\binom{2^{L'}-(M-1)Q+M-1}{M-1}\right] $. Note that $\binom{2^{L'}-(M-1)Q+M-1}{M-1}\ge \lceil\frac{(2^{L'}-MQ)^{M-1}}{(M-1)!}\rceil$. Let us assume for now that the mapping~$F^H_S$ and its inverse $(F^H_S)^{-1}$ are given. 

In the following, we provide more intuition regarding item (2) described above. We first give the definition of an indicator vector mentioned above and show that a $K$-substitution error results in at most $2K$ subsititution errors in the indicator vector, which can be corrected using classic Reed-Solomon codes. Note that the set of strings $\{(x_{i,1},x_{i,2},\ldots,x_{i,L'})\}^M_{i=1}$ can be determined uniquely by its indicator vector. 

For a set~$S\in \binom{\{0,1\}^{L'}}{\le M}$, define the indicator vector~$\1(S)\in\{0,1\}^{2^{L'}}$ of~$S$ by
\begin{align*}
    		\1 (S)_i=\begin{cases}
			1 & \mbox{if the binary presentation of~$i-1$ is in~$S$}\\
			0 &\mbox{else} 
		\end{cases}
\end{align*}
for $i\in[2^{L'}]$. 
Notice that the Hamming weight of~$\1(S)$ is~$M$ for every~$S\in\binom{\{0,1\}^{L'}}{ M}$, 
 and the following simple lemma holds.
\begin{lemma}\label{lemma:distance1S}
	For~$S_1,S_2\in \binom{\{0,1\}^{L'}}{\le M}$, if~$S_1\in\cB^H_K(S_2)$, then~$d_H(\1(S_1),\1(S_2))\le 2K$, where~$d_H(\1(S_1),\1(S_2))$ is the Hamming distance between~$\1(S_1)$ and~$\1(S_2)$
\end{lemma}
\begin{proof}
	Note that~$|S_1\backslash S_2|\le K$ and~$|S_2\backslash S_1|\le K$, where $S_1\backslash S_2=\{\boldx:\boldx\in S_1,\boldx\notin S_2\}$.
	Hence,~$d_H(\1(S_1),\1(S_2))=|(S_1\backslash S_2)\cup( S_2\backslash S_1)| \le 2K$.
\end{proof}
We now turn to present the code construction.
We use a set~$S\in\cS^H$ as indexing bits and protect its indicator vector~$\1(S)$ from substitution errors. Note that any two strings in the set~$S$ have Hamming distance at least~$2K+1$. Hence, if the set~$S$ is known, each string of indexing bits $(x_{i,1},\ldots,x_{i,L'})$ can be extracted from its erroneous version $(x'_{i,1},\ldots,x'_{i,L'})$ using a minimum distance decoder, which finds the unique string in~$S$ that is within Hamming distance~$K$ from $(x'_{i,1},\ldots,x'_{i,L'})$.  
The details are given as follows.

Let the message to be encoded be represented as a tuple~$\boldd=(d_1,\boldd_2)$, where~$d_1\in [\binom{2^{L'}-(M-1)Q+M-1}{M-1}]$ and 
$$\boldd_2\in \{0,1\}^{M(L-L')-4KL'-2K\lceil\log ML\rceil}.$$
Given~$(d_1,\boldd_2)$, the codeword~$\{\boldx_i\}^M_{i=1}$ is generated by the following procedure.

\textbf{Encoding:} 
\begin{itemize}
    \item [\textbf{(1)}] Let~$F^H_S(d_1)=\{\bolda_1,\ldots,\bolda_M\}\in\mathcal{S}^H$ such that~$\bolda_1=\1_{L'}$ and the $\bolda_i$'s are sorted in a descending lexicographic order, i.e., 
    $\bolda_i>\bolda_j$  for~$i<j$. 
    Let $(x_{i,1},\ldots,x_{i,L'})=\bolda_i$, for~$i\in[M]$.
    \item [\textbf{(2)}] Let $(x_{1,L'+1},\ldots,x_{1,L'+4KL'})=RS_{2K}(\1 (\{\bolda_1,\ldots,\bolda_M\}))$, where~$RS_{2K}(\1 (\{\bolda_1,\ldots,\bolda_M\}))\in\{0,1\}^{4kL'}$ is the redundancy of a systematic Reed-Solomon code that corrects~$2K$ substitutions in~$\1 (\{\bolda_1,\ldots,\bolda_M\})$, i.e., $(\1 (\{\bolda_1,\ldots,\bolda_M\}),$\newline$RS_{2K}(\1 (\{\bolda_1,\ldots,\bolda_M\})))$ can be recovered from $2K$ substitution errors. 
    \item [\textbf{(3)}] Place the information bits of~$\boldd_2$ in bits
    \begin{align*}
        &(x_{1,L'+4KL'+1},\ldots,x_{1,L}),\\ &(x_{i,L'+1},\ldots,x_{i,L})\mbox{ for }i\in[2,M-1]\mbox{, and}\\ &
        (x_{M,L'+1},\ldots,x_{M,L-2K\lceil\log ML\rceil}).
    \end{align*}
    \item [\textbf{(4)}] Define
    \begin{align}\label{eq:m}
        \boldm = (\boldx_1,\ldots,\boldx_{M-1},(x_{M,1},\ldots,x_{M,L-2K\lceil\log ML\rceil}))
    \end{align}
    and let $(x_{M,L-2K\lceil\log ML\rceil+1},\ldots,x_{M,L})=RS_K(\boldm)$, which is the Reed-Solomon redundancy that corrects~$K$ substitution errors in~$\boldm$. 
\end{itemize}  
Upon receiving the erroneous version\footnote{Since the strings among~$\{\boldx_i\}^M_{i=1}$ have distance at least~$2K+1$ to each other, the strings~$\{\boldx'_i\}^M_{i=1}$ are different.}~$\{\boldx'_1,\ldots,\boldx'_M\}$, the decoding procedure is as follows. 

\textbf{Decoding:} 
\begin{itemize}
    \item [\textbf{(1)}] Note that during the encoding process, the redundancy bits needed to correct the vector~$\1(\{\bolda_i\}^M_{i=1})$ from $2K$ substitutions are stored in~$\boldx_1$. Hence, we first identify the erroneous copy of~$\boldx_1$.
    To this end, find the unique string~$\boldx'_{i_0}$ such that the number of $1$-entries in~$(x'_{i_0,1},\ldots,$ $x'_{i_0,L'})$ is at least~$L'-K$. Since the strings in~$\{(x_{i,1},\ldots,x_{i,L'})\}^M_{i=1}$ have Hamming distance at least~$2K+1$, there is a unique such string, which is the erroneous copy  of~$(x_{1,1},\ldots,x_{1,L'})=\1_{L'}$. Hence, $\boldx'_{i_0}$ is an erroneous copy of~$\boldx_1$ and the string~$$(x'_{i_0,L'+1},\ldots,x'_{i_0,L'+4KL})$$ is an erroneous copy of~$(x_{1,L'+1},\ldots, $ $x_{1,L'+4KL'})=RS_{2K}(\1(\{\bolda_i\}^M_{i=1}))$. 
    \item [\textbf{(2)}] Let~$K'\le K$ be the number of substitution errors that occur in the indexing bits~$\{x_{i,1},\ldots,x_{i,L'}\}^M_{i=1}$. According to Lemma~\ref{lemma:distance1S}, the vector $\1 (\{(x_{i,1},\ldots,x_{i,L'})\}^M_{i=1})$ is within Hamming distance~$2K'$ from the vector~$\1 (\{(x'_{i,1},\ldots,x'_{i,L'})\}^M_{i=1})$. Hence, the Hamming distance between
    \begin{align*}
    &\bolds_1=(\1(\{(x'_{i,1},\ldots,x'_{i,L'})\}^M_{i=1}),(x'_{i_0,L'+1},\ldots,x'_{i_0,L'+4KL}))\mbox{ and}\\
    &\bolds_2=(\1(\{\bolda_i\}^M_{i=1}),RS_{2K}(\1(\{\bolda_i\}^M_{i=1})))	
    \end{align*} 
    is at most~$2K'+K-K'\le 2K$. Since~$\bolds_2$ is a codeword from a Reed-Solomon code of minimum distance~$2K+1$, it can be recovered from~$\bolds_1$ using any Reed-Solomon decoder. Hence, apply a Reed-Solomon decoder to obtain $\bolds_2$ from $\bolds_1$, extract the set $\{\bolda_i\}^M_{i=1}$, and feed it into the inverse mapping $(F^H_S)^{-1}$ (See Lemma \ref{theorem:polytimefh}) to decode~$d_1=(F^H_S)^{-1}(\{\bolda_i\}^M_{i=1})$.
    \item [\textbf{(3)}] Since~$\bolds_2$ is recovered, the strings~$\{(x_{i,1},\ldots,x_{i,L'})\}^M_{i=1}=\{\bolda_i\}^M_{i=1}$ are known. Sort~$\{(x_{i,1},\ldots,x_{i,L'})\}^M_{i=1}$ lexicographically in descending order.
    For each~$i\in[M]$, find the unique~$\pi(i)\in [M]$ such that~$d_H((x'_{\pi(i),1},\ldots,x'_{\pi(i),L'}),(x_{i,1},\ldots,x_{i,L'}))\le K$  (note that~$i_0=\pi(1)$). Similar to Step~(1), we conclude that the string~$\boldx'_{\pi(i)}$ is an erroneous copy of~$\boldx_i$,~$i\in [M]$, since the Hamming distance between~$(x_{j,1},\ldots,x_{j,L'})$ and~$(x_{i,1},\ldots,x_{i,L'})$ is at least~$2K+1$ for~$j\ne i$. Hence, the identify of ~$\{(x_{i,1},\ldots,x_{i,L'})\}^M_{i=1}$ are determined from~$\{(x'_{i,1},\ldots,x'_{i,L'})\}^M_{i=1}$.
    \item [\textbf{(4)}] Since~$\boldx'_{\pi(i)}$ is an erroneous copy of~$\boldx_i$,~$i\in[M]$.
    it follows that the concatenation~$\bolds'=(\boldx'_{\pi(1)},\ldots,\boldx'_{\pi(M)})$ is an erroneous copy of $(\boldx_{1},\ldots,\boldx_{M})=(\boldm,RS_K(\boldm))$, where~$\boldm$ is defined in \eqref{eq:m}. Note that there are at most $K$ substitution errors in $\bolds'$.  Therefore,~$(\boldx_{1},\ldots,\boldx_{M})$, and thus~$\boldd_2$,  can be recovered from $(\boldx'_{\pi(1)},\ldots,\boldx'_{\pi(M)})$ by using a Reed-Solomon decoder. \item [\textbf{(5)}] Output~$(d_1,\boldd_2)$.
\end{itemize}
Therefore, the codeword~$\{\boldx_i\}^M_{i=1}$ can be recovered. 
The redundancy of the code is 
\color{black}
\begin{align}\label{equation:inequality2}
    r(\mathcal{C})
    =&\log \binom{2^L}{M}-\log \binom{2^{L'}-(M-1)Q+M-1}{M-1}- [M(L-L')-4KL'-2K\lceil\log ML\rceil]\\
    \overset{(a)}{\le} &2K\log ML + (12K+2)\log M+O(K^3)+O(K\log\log ML),
\end{align}
where~$(a)$ is proved in Appendix~\ref{section:proofofinequality}.
The complexity of the encoding/decoding is that of computing the function~$F^H_S$, which is~$poly(M,L,K)$, as will be discussed in Section~\ref{section:computingfh}.

\section{Computing~$F^H_S$ in polynomial time}\label{section:computingfh}
In this section we present a polynomial time algorithm to compute the function~$F^H_S$ and thus complete the code construction in Section~\ref{section:robusti}. The result is as follows.
\begin{lemma}\label{theorem:polytimefh}
	For integers~$M,L,K$,~$L'\triangleq 3\log M  + 4K^2+2$ and~$Q=\sum^{2K}_{i=0}\binom{L'}{i}$, there exists a  mapping~$F^H_S: \left[\binom{2^{L'}-(M-1)Q+M-1}{M-1}\right]\rightarrow \binom{\{0,1\}^{L'}}{M}$, computable in~$poly(M,L)$ time, such that for any~$d\in \left[\binom{2^{L'}-(M-1)Q+M-1}{M-1}\right]$, we have that~$F^H_S(d)\in \cS^H$. 
 In addition, there exists a map $(F^H_S)^{-1}$, computable in $poly(M,L)$ time, such that given any $F^H_S(d)$ for some $d\in \left[\binom{2^{L'}-(M-1)Q+M-1}{M-1}\right]$, we have that $(F^H_S)^{-1}(F^H_S(d))=d$.
\end{lemma} 
The algorithm for computing $F^H_S$  consists of two steps. In the first step,  we map the integer~$d\in [\binom{2^{L'}-(M-1)Q+M-1}{M-1}]$ into a set of~$M$ different  integers~$q_1,\ldots,q_{M}\in[2^{L'}]$ such that~$q_1=2^{L'}$ and~$q_{i+1}\le q_i-Q$ for~$i\in[M-1]$. In the second step,
 we use~$q_i$ and a greedy algorithm  to generate~$\bolda_i$ sequentially for~$i\in[2,M]$, where each string~$\bolda_i$,~$i\in [2,M]$ is generated bit by bit. The first step is given in the following lemma.
\begin{lemma}\label{lemma:dtoq}
	There exists a   map~$F^H_{Q}:[\binom{2^{L'}-MQ+M-1}{M-1}]\rightarrow \binom{[2^{L'}]}{M}$, computable in~$poly(L',M)$ time, that maps an  integer~$d\in [\binom{2^{L'}-MQ+M-1}{M-1}]$ to an integer set~$\{q_1,\ldots,q_{M}\}\subseteq [2^{L'}]$ such that~$q_1=2^{L'}$,~$q_{i}\ge q_{i+1}+Q$ for~$i\in[M-1]$, and $q_M\ge Q$. In addition, there exists a map $(F^H_Q)^{-1}$, computable in $poly(L',M)$ time, such that for every $d\in[\binom{2^{L'}-MQ+M-1}{M-1}]$, we have that $(F^H_Q)^{-1}(F^H_Q(d))=d$.
\end{lemma} 
\begin{proof}
	Recall the combinatorial numbering map~$F_{com}$ that maps an integer in the range~$[\binom{n}{m}]$ to a set of~$m$ different and unordered integers in the range~$[n]$ for integers~$n$ and~$m\le n$.
	We map~$d\in [\binom{2^{L'}-MQ+M-1}{M-1}]$ to~$M-1$ different integers~$F_{com}(d)=\{q'_2,\ldots,q'_{M}\}$ such that~$2^{L'}-MQ+M-1\ge q'_2>q'_3>\ldots>q'_{M}$.
	Let~$q_1=2^{L'}$,~$q_i=q'_i+(M-i+1)(Q-1)$ for~$i\in[2,M]$, and~$F^H_Q(d)=\{q_1,\ldots,q_M\}$. Then we have that~$q_2\le 2^{L'}-Q$ and that~$q_{i}\ge q_{i+1}+Q$ for~$i\in [2,M-1]$. 
	
	To compute $(F^H_Q)^{-1}(\{q_1,\ldots,q_{M}\})$, where $\{q_1,\ldots,q_M\}=F^H_Q(d)$, we first compute $q'_i=q_i-(M-i+1)(Q-1)$ for $i\in[2,M]$. Since $q_1=2^{L'}$ and $q_{i+1}\le q_i-Q$, we have $q'_M<q'_{M_1}<\ldots<q'_2\le 2^{L'}-MQ+M-1$. 
	Note that the map~$F_{com}$ is invertible and computable in~$poly(L',M)$ time. One can compute $d=F^{-1}_{com}(\{q'_2,\ldots,q'_M\})$.
\end{proof}
We now turn to the second step. Given  integers~$F^H_Q(d)=\{q_1,\ldots,q_M\}$, we generate the indexing bits   $\{\bolda_i=(x_{i,1},\ldots,$ $x_{i,L'})\}^M_{i=1}\in \cS^H$. First, we have that~$\bolda_1=\1_{L'}$. We then generate the indexing string~$\bolda_i$ sequentially for~$i\in[2,M]$. Each~$\bolda_i$ is generated bit by bit in a recursive manner. The following definition is used throughout the algorithm.

For a set of strings~$A\subset \{0,1\}^{L'}$ and a string~$\bolda\in\{0,1\}^\ell$ of length~$\ell\in[L']$. Denote
\begin{align*}
	N_H(\bolda,A)=\sum_{\boldc:\boldc\in A}|\{\boldc':\boldc'\in\{0,1\}^{L'},~(c'_1,\ldots,c'_{\ell})=\bolda\mbox{ and }d_H(\boldc',\boldc)\le 2K\}|
\end{align*}
 as the sum over all $\boldc\in A$ of the number of sequences of length $L'$ that have prefix~$\bolda$ and have Hamming distance at most~$2K$ from~$\boldc$. The number~$N_H(\bolda,A)$ has the following properties (stated in Lemma \ref{lemma:twoproperties}) that are useful in our proof. The first property  
enables a recursion to generate each sequence~$\bolda_i$. 
 The second property provides a way to compute~$N_H(\bolda,A)$.
 \begin{lemma}\label{lemma:twoproperties}
  \begin{enumerate}
 	\item For any sequence~$\bolda\in\{0,1\}^\ell$ of length~$\ell\in[L'-1]$ and set~$A\subset\{0,1\}^{L'}$, we have 
 \begin{align}\label{equation:recursion}
 2^{L'-\ell}-N_H(\bolda,A)=(2^{L'-\ell-1}-N_H((\bolda,0),A))+(2^{L'-\ell-1}-N_H((\bolda,1),A))
 \end{align}	
 	where~$(\bolda,0)$ or~$(\bolda,1)$ is the concatenation of~$\bolda$ and a~$0$ or~$1$ bit respectively.
 	\item For any~$\bolda\in\{0,1\}^{\ell}$ and~$A\subset\{0,1\}^{L'}$, we have
 	\begin{align}\label{equation:computeN}
 	N_H(\bolda,A)=\sum_{\boldc:\boldc\in A}\sum^{2K-d_H(\bolda,(c_1,\ldots,c_\ell))}_{i=0}\binom{L'-\ell}{i}.
 	\end{align}
 \end{enumerate}	
 \end{lemma}
\begin{proof}
	 Note that for any sequence~$\boldc$, the~$(\ell+1)$-th bit of any sequence~$\boldc'$ satisfying~$(c'_1,\ldots,c'_{\ell})=\bolda$ is either~$0$ or~$1$. Hence 
	\begin{align*}
	&|\{\boldc':\boldc'\in\{0,1\}^{L'},~(c'_1,\ldots,c'_{\ell})=\bolda\mbox{ and }d_H(\boldc',\boldc)\le 2K\}|\\
	= &|\{\boldc':\boldc'\in\{0,1\}^{L'},~(c'_1,\ldots,c'_{\ell+1})=(\bolda,0)\mbox{ and }d_H(\boldc',\boldc)\le 2K\}|\\
	&+ |\{\boldc':\boldc'\in\{0,1\}^{L'},~(c'_1,\ldots,c'_{\ell+1})=(\bolda,1)\mbox{ and }d_H(\boldc',\boldc)\le 2K\}|,
	\end{align*}
	which implies
	Eq.~\eqref{equation:recursion} by definition of $N_H(\bolda,A)$.
	Moreover,
	for any sequence~$\boldc\in\{0,1\}^{L'}$, we have that
	\begin{align*}
	|\{\boldc':\boldc'\in\{0,1\}^{L'},~(c'_1,\ldots,c'_{\ell})=\bolda\mbox{ and }d_H(\boldc',\boldc)\le 2K\}|=\sum^{2K-d_H(\bolda,(c_1,\ldots,c_\ell))}_{i=0}\binom{L'-\ell}{i}
	\end{align*}
	Hence the number~$N_H(\bolda,A)$ can be computed by Eq.~\eqref{equation:computeN}.
\end{proof}

Next, we present the algorithm (Algorithm \ref{alg:ptoasubstitution}) that takes~$F^H_Q(d)=\{q_1,\ldots,q_M\}$ as input and outputs~$\bolda_i$ such that $\{\bolda_1,\ldots,\bolda_M\}\in\cS^H$ and that the decimal presentation~$\decimal(\bolda_i)$ of~$\bolda_i,i\in[M]$ satisfies
\begin{align}\label{equation:ai}
\decimal(\bolda_i)=q_i-1+\sum_{\ell:a_{i,\ell}=1} N_H((a_{i,1},\ldots,a_{i,\ell-1},0),\{\bolda_j\}^{i-1}_{j=1}).
\end{align}
We then show that the sequences~$\bolda_i,i\in[M]$ satisfying~\eqref{equation:ai} are decodable, i.e., we can recover~the set~$\{q_1,\ldots,q_M\}$ from~$\{\bolda_1,\ldots,\bolda_M\}$. 

%
%
%
%
%
%
%
\begin{algorithm}
  \caption{Encoding from $\{q_1,\ldots,q_M\}$ to $\{\bolda_1,\ldots,\bolda_M\}$}\label{alg:ptoasubstitution}
  \begin{algorithmic}[1]
      \For{$i\in [M]$}
        \State $q=q_i$
        \For{$\ell\in[L']$}
            \If{$2^{L'-\ell}-N_H((a_{i,1},\ldots,a_{i,\ell-1},0),\{\bolda_j\}^{i-1}_{j=1})\ge q$}
            \State $a_{i,\ell}=0$
            \Else
            \State $q=q-(2^{L'-\ell}-N_H((a_{i,1},\ldots,a_{i,\ell-1},0),\{\bolda_j\}^{i-1}_{j=1}))$
            \State $a_{i,\ell}=1$
            \EndIf
        \EndFor
      \EndFor
      \State \textbf{return} $\{\bolda_1,\ldots,\bolda_M\}$
  \end{algorithmic}
\end{algorithm}
The generation of~$\bolda_i,i\in[M]$ in Algorithm \ref{alg:ptoasubstitution} can be intuitively characterized as traversing a complete binary tree of~$L'+1$ layers. The walk starts at layer~$1$, i.e., the root of the binary tree, and ends at layer~$L'+1$ at one of the leaf nodes. At each step, it goes to one of its two child nodes, which represent the bits~$0$ and~$1$ respectively. Each string~$\bolda_i,i\in[M]$ is represented by the path of a walk. For each path~$\bolda_i=(a_{i,1},\ldots,a_{i,L'})$ and each layer~$\ell\in[L']$, assign the weight~$w(a_{i,\ell})=2^{L'-\ell}-N_H((a_{i,1},\ldots,a_{i,\ell-1},a_{i,\ell}),\{\bolda_j\}^{i-1}_{j=1})$ to node~$a_{i,\ell}$ in the~$\ell$-th layer, and the weight~$w(\bar{a}_{i,\ell})=2^{L'-\ell}-N_H((a_{i,1},\ldots,a_{i,\ell-1},1-a_{i,\ell}),\{\bolda_j\}^{i-1}_{j=1})$ to the sibling of node~$a_{i,\ell}$, i.e., the node that shares the same parent node with~$a_{i,\ell}$. From Eq.~\eqref{equation:recursion} we have that~$w(a_{i,\ell})=w(a_{i,\ell+1})+w(\bar{a}_{i,\ell+1})$ for~$\ell\in[L'-1]$. Moreover, we have that~$0<q\le w(a_{i,\ell})$ after the~$\ell$-th inner for loop in the~$i$-th outer for loop.
This is formalized in the following lemma, which can be used to prove that Eq.~\eqref{equation:ai} holds and that~$\{\bolda_1,\ldots,\bolda_M\}\in \cS^H$. 
\begin{lemma}\label{lemma:qequals1}
	For $\ell\in[L']$ and $i\in[M]$, after the~$\ell$-th inner for loop in the~$i$-th outer for loop in Algorithm \ref{alg:ptoasubstitution}, we have that
	\begin{align}\label{equation:qdynamics}
		0<q\le 2^{L'-\ell}-N_H((a_{i,1},\ldots,a_{i,\ell}),\{\bolda_j\}^{i-1}_{j=1}).
	\end{align} 
	In addition, at the end of the~$i$-th outer for loop, we have that~$q=1$.
\end{lemma}
\begin{proof}
	We prove Eq.~\eqref{equation:qdynamics} by induction on~$\ell$. For~$\ell=1$, according to Lemma~\ref{lemma:dtoq}, we have~$0<q=q_i\le 2^{L'}-(i-1)Q$ at the beginning of the~$i$-th outer for loop. Eq. \eqref{equation:qdynamics} holds because $N_H(\emptyset,\{\bolda_j\}^{i-1}_{j=1})=(i-1)Q$, where $\emptyset$ is an empty string. If~$a_{i,1}=0$, then according to the if condition in Algorithm \ref{alg:ptoasubstitution}, we have that~$0<q\le 2^{L'-\ell}-N_H(0,\{\bolda_j\}^{i-1}_{j=1})$ for~$\ell=1$, which proves~\eqref{equation:qdynamics}. Otherwise if~$a_{i,1}=1$, we have
	\begin{align*}
	0<q=&\;q_i- (2^{L'-1}-N_H(0,\{\bolda_j\}^{i-1}_{j=1}))\\
	\le &  \;	2^{L'}-(i-1)Q-(2^{L'-1}-N_H(0,\{\bolda_j\}^{i-1}_{j=1}))\\
	\overset{(a)}{=}&(2^{L'-1}-N_H(1,\{\bolda_j\}^{i-1}_{j=1}))
	\end{align*}
	where~$(a)$ follows from \eqref{equation:recursion} and the fact that $N_H(\emptyset,\{\bolda_j\}^{i-1}_{j=1})=(i-1)Q$~$N_H(\bolda,A)$.  
	Hence the claim holds for~$\ell=1$. Suppose Eq.~\eqref{equation:qdynamics} holds for~$\ell=m$. For~$\ell=m+1$, if~$a_{i,m+1}=0$, then from the if condition in the $\ell$-th inner loop, we have~$0<q\le 2^{L'-m-1}-N_H((a_{i,1},\ldots,a_{i,m},0),\{\bolda_j\}^{i-1}_{j=1})$. Otherwise if~$a_{i,m+1}=1$, we have that
	\begin{align*}
	0<q=&\;q_i- (2^{L'-m-1}-N_H((a_{i,1},\ldots,a_{i,\ell},0),\{\bolda_j\}^{i-1}_{j=1}))\\
\le & \; 	2^{L'-m}-N_H((a_{i,1},\ldots,a_{i,m}),\{\bolda_j\}^{i-1}_{j=1})-(2^{L'-m-1}-N_H((a_{i,1},\ldots,a_{i,m},0),\{\bolda_j\}^{i-1}_{j=1}))\\
\overset{(b)}{=}&\;(2^{L'-m-1}-N_H((a_{i,1},\ldots,a_{i,m},1),\{\bolda_j\}^{i-1}_{j=1})),		
	\end{align*}
	where~$(b)$ follows from Eq.~\eqref{equation:recursion}. Therefore, Eq.~\eqref{equation:qdynamics} holds for~$\ell=m+1$ and thus holds for~$\ell\in[L']$.
	Hence at the end of the $L'$-th inner loop in the $i$-th outer loop, we have that
	\begin{align}\label{equation:qequals1}
				0<q\le 2^{L'-L'}-N_H(\bolda_i,\{\bolda_j\}^{i-1}_{j=1})
				\le 1.
	\end{align}
	Hence~$q$ equals~$1$ at the end of the $L'$-th inner loop in the $i$-th outer loop.
\end{proof}
We now show that the strings~$\{\bolda_1,\ldots,\bolda_M\}$ generated in Algorithm \ref{alg:ptoasubstitution} belong to $\cS_H$. By Lemma~\ref{lemma:qequals1}, we have
\begin{align*}
q=2^{L'-L'}-N_H(\bolda_i,\{\bolda_j\}^{i-1}_{j=1})=1,
\end{align*}
at the end of the $i$-th outer for loop in Algorithm \ref{alg:ptoasubstitution}. This implies that~$N_H(\bolda_i,\{\bolda_j\}^{i-1}_{j=1})=0$ and thus~$d_H(\bolda_i,\bolda_j)\ge 2K+1$ for~$i\in[2,M]$ and~$j\in[i-1]$. Moreover, since~$q_1=2^{L'}$, we have that~$\bolda_1=\1_{L'}$, because $N_H(\bolda_i,\{\bolda_j\}^{i-1}_{j=1})=0$ for $i=1$ and the if condition in the $\ell$-th loop is always not satisfied. Therefore,~$\{\bolda_i\}^M_{i=1}\in\cS_H$. 

Next, we use Lemma~\ref{lemma:qequals1} to show that the strings~$\{\bolda_i\}^M_{i=1}$ satisfy Eq.~\eqref{equation:ai}. 
\begin{lemma}\label{lemma:ai}
	The output~$\{\bolda_i\}^M_{i=1}$ of the encoding algorithm satisfies Eq.~\eqref{equation:ai}. 
\end{lemma}
\begin{proof}
	Note that in each inner for loop, the number~$q$ is reduced  by~$2^{L'-\ell}-N_H((a_{i,1},\ldots,a_{i,\ell-1},0),\{\bolda_j\}^{i-1}_{j=1})$ only when~$a_{i,\ell}=1$ and~$\ell\in[L']$. Since the number~$q$ equals~$q_i$ at the beginning of each outer for loop, and from  Lemma~\ref{lemma:qequals1} equals~$1$ at the end of each outer for loop, hence we have that
	\begin{align*}
		q_i-\sum_{\ell:a_{i,\ell}=1\mbox{ and }\ell\in[L']}(2^{L'-\ell}- N_H((a_{i,1},\ldots,a_{i,\ell-1},0),\{\bolda_j\}^{i-1}_{j=1}))=1
		,
	\end{align*}
	which implies~\eqref{equation:ai}.
\end{proof}
\begin{remark}
	By \eqref{equation:ai} and the definition of~$N_H((a_{i,1},\ldots,a_{i,\ell-1},0),\{\bolda_j\}^{i-1}_{j=1})$, we have the following alternative characterization of~$\decimal(\bolda_i)$,~$i\in[M]$.
	\begin{align}\label{equation:aialternative}
	\decimal(\bolda_i)=q_i-1+\sum_{j=1}^{i-1} |\{\boldc: decimal(\boldc) < decimal(\bolda_i)\mbox{ and }d_H(\boldc,\bolda_j)\le 2K \}|,
	\end{align}
	which is~$q_i-1$ plus the sum of number of strings that are lexicographically less than~$\bolda_i$ and have Hamming distance at most~$2K$ from~$\bolda_j$ over~$j<i$.
\end{remark}
Lemma~\ref{lemma:ai} immediately implies a decoding algorithm that transforms~$\{\bolda_i\}^M_{i=1}$ back to~$\{q_1,\ldots,q_M\}$.

\textbf{Decoding:} 
\begin{itemize}
	\item [\textbf{(1)}] Order the strings~$\{\bolda_{i}\}^{M}_{i=1}$ such that~$\bolda_1>\bolda_2>\ldots>\bolda_M$.
	\item [\textbf{(2)}]
	For~$i\in[M]$,
	\begin{align}\label{equation:decodingqisubstitution}
	q_i=\decimal(\bolda_i)+1+\sum_{\ell:a_{i,\ell}=1\mbox{ and }\ell\in[L']}N_H((a_{i,1},\ldots,a_{i,\ell-1},0),\{\bolda_j\}^{i-1}_{j=1}).
	\end{align}
\end{itemize}  
To show that the decoding is correct, we prove that the strings~$\bolda_i$,~$i\in[M]$ generated in Algorithm \ref{alg:ptoasubstitution} satisfy
\begin{align}\label{equation:orderai}
	\bolda_1>\bolda_2>\ldots>\bolda_M.
\end{align}
Then we conclude that the string~$\bolda_i$ obtained by ordering~$\{\bolda_i\}^M_{i=1}$ in Step (1) in the decoding procedure satisfies Eq.~\eqref{equation:ai}. Hence we have
Eq.~\eqref{equation:decodingqisubstitution} and thus~$q_i$,~$i\in[M]$ can be recovered. Suppose to the contrary, there exist~$\bolda_{i_1}>\bolda_{i_2}$ for some~$i_1>i_2$. Let~$\ell^*$ be the most significant bit where~$\bolda_{i_1}$ and~$\bolda_{i_2}$ differ, i.e.,~$(a_{i_1,1},\ldots,a_{i_1,\ell^*-1})=(a_{i_2,1},\ldots,a_{i_2,\ell^*-1})$ and~$a_{i_1,\ell^*}=1$ and~$a_{i_2,\ell^*}=0$. Then according to the if statement in Algorithm \ref{alg:ptoasubstitution}, we have that
\begin{align*}
	q_{i_1}-\sum_{\ell:a_{i_1,\ell}=1\mbox{ and }\ell\in[\ell^*]}(2^{L'-\ell}- N_H((a_{i_1,1},\ldots,a_{i_1,\ell-1},0),\{\bolda_j\}^{i_1-1}_{j=1}))&>0~\mbox{ and}\\
	q_{i_2}-\sum_{\ell:a_{i_1,\ell}=1\mbox{ and }\ell\in[\ell*]}(2^{L'-\ell}- N_H((a_{i_1,1},\ldots,a_{i_1,\ell-1},0),\{\bolda_j\}^{i_2-1}_{j=1}))&\le 0,
\end{align*}
which implies that
\begin{align}\label{equation:qi2qi1}
	q_{i_2}-q_{i_1}<& \sum_{\ell:a_{i_1,\ell}=1\mbox{ and }\ell\in[\ell^*]}(2^{L'-\ell}- N_H((a_{i_1,1},\ldots,a_{i_1,\ell-1},0),\{\bolda_j\}^{i_2-1}_{j=1}))\nonumber\\
	&-\sum_{\ell:a_{i_1,\ell}=1\mbox{ and }\ell\in[\ell^*]}(2^{L'-\ell}- N_H((a_{i_1,1},\ldots,a_{i_1,\ell-1},0),\{\bolda_j\}^{i_1-1}_{j=1}))\nonumber\\
	=&\sum_{\ell:a_{i_1,\ell}=1\mbox{ and }\ell\in[\ell^*]}(N_H((a_{i_1,1},\ldots,a_{i_1,\ell-1},0),\{\bolda_j\}^{i_1-1}_{j=1})-N_H((a_{i_1,1},\ldots,a_{i_1,\ell-1},0),\{\bolda_j\}^{i_2-1}_{j=1}))\nonumber\\
	=&\sum_{\ell:a_{i_1,\ell}=1\mbox{ and }\ell\in[\ell^*]}N_H((a_{i_1,1},\ldots,a_{i_1,\ell-1},0),\{\bolda_j\}^{i_1-1}_{j=i_2})\nonumber\\
	\overset{(a)}{\le}& \sum^{i_1-1}_{j=i_2}|\boldc:d_H(\boldc,\bolda_j)\le 2K|\nonumber\\
	=&(i_1-i_2)Q,
\end{align}
where~$(a)$ follows from the definition of~$N_H(\bolda,A)$ and the fact that the strings which have $(a_{i_1,1},\ldots,a_{i_1,\ell_1-1},0)$ and $(a_{i_1,1},\ldots,$ $a_{i_1,\ell_2-1},0)$ as prefixes, respectively, are different for~$a_{i_1,\ell_1}=1$,~$a_{i_1,\ell_2}=1$ and~$\ell_1\ne \ell_2$. Eq.~\eqref{equation:qi2qi1} contradicts to the fact that the integers~$\{q_1,\ldots,q_M\}=F^H_Q(d)$ satisfy~$q_{i}-q_{i+1}\ge Q$ for~$i\in[M-1]$, which implies~$q_{i_2}-q_{i_1}\ge (i_1-i_2)Q$.

	Since the calculation of~$N_H(\bolda,A)$ follows \eqref{equation:computeN} and has polynomial complexity, the complexity of the encoding/decoding procedure is polynomial in~$M$ and~$L'$. 
	
	Finally, to compute $(F^H_S)^{-1}(F^H_S(d))$ given $F^H_S(d)=\{\bolda_1,\ldots,\bolda_M\}$, we first use the decoding procedure described above to recover $\{q_1,\ldots,q_M\}$. Then we obtain $(F^H_Q)^{-1}(\{q_1,\ldots,q_{M}\})=d$, where $(F^H_Q)^{-1}$ is defined in Lemma \ref{lemma:dtoq}.

\section{Robust indexing for deletion/insertion errors}\label{section:computingfd}
In this section we show how the idea of robust indexing can be used for correcting deletion/insertion errors over unordered set of strings. The redundancy of the construction is~$O(K\log ML)$ for constant~$K$, which is order-wise optimal with respect to $M$ and $L$. 

Similar to the construction in Section~\ref{section:robusti}, 
we use the first~$L'$ bits~$(x_{i,1},\ldots,x_{i,L'})$,~$i\in[M]$ in each string~$\boldx_i$ as indexing bits and sort the strings~$\{\boldx_i\}^M_{i=1}$ according to the lexicographic order of~$\{(x_{i,1},\ldots,x_{i,L'})\}^M_{i=1}$. To protect the ordering, we use Reed-Solomon code redundancy to protect the indicator vector~$\1  (\{(x_{i,1},\ldots,x_{i,L'})\}^M_{i=1})$. Then, we need deletion correcting codes to protect the Reed-Solomon code redundancy from deletion errors. The difference between the schemes for deletion/insertion errors and for substitution errors is that for correcting deletion/insertion errors, we construct the indexing bits~$\{(x_{i,1},\ldots,x_{i,L'})\}^M_{i=1}$ such that the mutual deletion distance among~$\{(x_{i,1},\ldots,x_{i,L'})\}^M_{i=1}$, rather than the mutual Hamming distance as considered in Section~\ref{section:robusti}, is at least~$2K+1$, i.e., the deletion balls~$\cD_K((x_{i,1},\ldots,x_{i,L'}))$ and~$\cD_K((x_{j,1},\ldots,x_{j,L'}))$ do not intersect for~$i\ne j$. For any binary string $\boldx\in\{0,1\}^m$, its deletion ball $\cD_K(\boldx)$ is the collection of all substrings of $\boldx$ of length at least $m-K$.
Define the set
\begin{align*}
	\cS^D=\left\{ \{\bolda_1,\ldots,\bolda_M\}:\cD_K(\bolda_i)\cap\cD_K(\bolda_j)=\emptyset\mbox{ for }i\ne j \right\}.
\end{align*}
Then $\cS^D$ is a code of size $M$ consisting of strings that are resilient to deletion errors in the classical setting. 
The construction is based on the following two lemmas, where the first one is robust indexing for deletion/insertion errors, which will be proved in Section~\ref{subsection:fds} and the second one is a deletion code construction, which appeared in \cite{kdeletion}.
\begin{lemma}\label{lemma:fds}
	    For~$P=2^K\binom{L'}{K}^2$, 
		there exists a mapping~$F^D_S: \left[\binom{2^{L'}-(M-1)P+M-1}{M-1}\right]\rightarrow \binom{\{0,1\}^{L'}}{M}$, computable in~$poly(M,L)$ time, such that for any~$d\in \left[\binom{2^{L'}-(M-1)P+M-1}{M-1}\right]$, we have that~$F^D_S(d)\in \cS^D$. In addition, there exists a mapping $(F^D_S)^{-1}$, computable in $poly(M,L)$ time, such that for every  $d\in \left[\binom{2^{L'}-(M-1)P+M-1}{M-1}\right]$, we have that $(F^D_S)^{-1}(F^D_S(d))=d$. 
\end{lemma}
\begin{lemma}\label{lemma:kdeletioncodes}
	    For~any integer~$n$ and~$N=n+4K\log n +o(\log n)$, there exists an encoding function~$Enc:\{0,1\}^n\rightarrow \{0,1\}^{N}$, computable in~$O(n^{2K+1})$ time, and a decoding function~$Dec:\{0,1\}^{N-K}\rightarrow \{0,1\}^{n}$, computable in~$O(n^{K+1})$ time, such that for any~$\boldc\in\{0,1\}^n$ and substring~$\boldd\in \{0,1\}^{N-K}$ of~$Enc(\boldc)$, we have that~$Dec(\boldd)=\boldc$. 
\end{lemma}
\subsection{Code constructions}
The code construction is the same as that in Section~\ref{section:robusti} except that in this section, the indexing bits~$\{(x_{i,1},\ldots,x_{i,L'})\}^M_{i=1}$ are generated using the map~$F^D_S$, the details of which will be given in Section~\ref{subsection:fds}. In addition, a deletion code in Lemma~\ref{lemma:kdeletioncodes} is used to protect the concatenated string. 

Let the data~$\boldd$ to be encoded be represented by a tuple~$\boldd=(d_1,\boldd_2)$, where~$d_1\in \left[\binom{2^{L'}-(M-1)P+M-1}{M-1}\right]$ and 
$\boldd_2\in \{0,1\}^{n}$ such that~$n+4K\log n+o(\log n)=M(L-L')-4KL'-4K\log (4KL')-o(\log L')$, which implies that~$n=M(L-L')-4KL'-4K\lceil\log ML\rceil-o(\log ML)$.
We present the encoding/decoding procedure as follows. 

\textbf{Encoding:} 
\begin{itemize}
	\item [\textbf{(1)}] Let~$F^D_S(d_1)=\{\bolda_1,\ldots,\bolda_M\}\in\mathcal{S}^H$ such that~$\bolda_1=\1_{L'}$ and $\bolda_1>\bolda_2>\ldots>\bolda_M$.
	Let $(x_{i,1},\ldots,x_{i,L'})=\bolda_i$, for~$i\in[M]$.
	\item [\textbf{(2)}] Let $(x_{1,L'+1},\ldots,x_{1,L'+\alpha_{RS} })=Enc(RS_{2K}(\1 (\{\bolda_1,\ldots,\bolda_M\})))$, where $\alpha_{RS}=4KL'+4K\log (4KL')+o(\log (4KL'))$ is the length of the $K$-deletion correcting code (Lemma \ref{lemma:kdeletioncodes}) that protects the redundancy of the Reed-Solomon code $RS_{2K}(\1 (\{\bolda_1,\ldots,\bolda_M\}))$.
	\item [\textbf{(3)}] Place the deletion code~$Enc(\boldd_2)$ in bits
	\begin{align*}
	&(x_{1,L'+\alpha_{RS}+1},\ldots,x_{1,L}),\mbox{ and}\\ &(x_{i,L'+1},\ldots,x_{i,L})\mbox{ for }i\in[2,M].
	\end{align*}
\end{itemize}  
Upon receiving~$\{\boldx'_i\}^M_{i=1}$, the decoding procedure is as follows.

\textbf{Decoding:} 
\begin{itemize}
	\item [\textbf{(1)}] 
	Find the unique string~$\boldx'_{i_0}$ such that~$(x'_{i_0,1},\ldots,x'_{i_0,L'-K})=\1_{L'-K}$. Note that since the erroneous version of $\boldx_1$ has at least $L'-K$ $1$-entries in the first $L'$ bits, and hence such $i_0$ exists. Moreover, since $\cD_K(\bolda_1)\cap\cD_K(\bolda_i)=\emptyset$ for $i\in[2,M]$, such  $i_0$ is unique.  Then,~$\boldx'_{i_0}$ is an erroneous copy of~$\boldx_1$ and the string~$$(x'_{i_0,L'+1},\ldots,x'_{i_0,L'+\alpha_{RS}-K}).$$ is an erroneous copy, or more precisely, a length $\alpha_{RS}-K$ subsequence  of~$(x_{1,L'+1},\ldots, $ $x_{1,L'+\alpha_{RS}})=Enc(RS_{2K}(\1(\{\bolda_i\}^M_{i=1})))$. By definition of $Enc$ and Lemma \ref{lemma:kdeletioncodes}, we can correct the vector~$RS_{2K}(\1(\{\bolda_i\}^M_{i=1}))$ and use it to recover~$\1(\{\bolda_i\}^M_{i=1})$. This is because $K$ deletions affect at most $K$ strings among $\{\bolda_i\}^M_{i=1}$ and thus at most $2K$ entries in $\1(\{\bolda_i\}^M_{i=1})$, similar to Lemma \ref{lemma:distance1S}. Therefore,  the indexing bits~$\{(x_{i,1},\ldots,x_{i,L'})\}^M_{i=1}$ can be recovered from $\{\bolda_i\}^M_{i=1}$. Recover~$d_1=(F^D_S)^{-1}(\{\bolda_i\}^M_{i=1})$. 
	\item [\textbf{(2)}] For each~$i\in[M]$, find the unique~$\pi(i)\in [M]$ such that~$(x'_{\pi(i),1},\ldots,x'_{\pi(i),L'-K})$ is a length~$L'-K$ substring of $(x_{i,1},\ldots,x_{i,L'})$ (note that~$\pi(1)=i_0$). Again, since $\cD_K(\bolda_i)\cap\cD_K(\bolda_j)=\emptyset$ for any $i\ne j$, such $\pi$ is unique. Checking if a string is a substring of another can be done in linear time using a greedy algorithm. 
	\item [\textbf{(3)}] Since~$\boldx'_{\pi(i)}$ is an erroneous copy of~$\boldx_i$,~$i\in[M]$,
	the concatenation 
	\begin{align}
	\boldm'=&((x'_{\pi(1),L'+\alpha_{RS}+1},\ldots,x'_{\pi(1),L_1}),\nonumber\\
	&(x'_{\pi(2),L'+1},\ldots,x'_{\pi(2),L_2}),\ldots,(x'_{\pi(M),L'+1},\ldots,x'_{\pi(M),L_M})),\label{equation:concatenation}
	\end{align}
	 where~$L_i$ is the length of~$\boldx'_{\pi(i)}$, $i\in[M]$, is a length at least $|Enc(\boldd_2)|-K$ subsequence of~$Enc(\boldd_2)$. Use the decoder~$Dec(\boldm')=\boldd_2$ to recover $\boldd_2$. 
	\item [\textbf{(4)}] Output~$(d_1,\boldd_2)$.
\end{itemize}
Similar to \eqref{equation:inequality2}, the redundancy of the code can be bounded by 
\begin{align*}
	 r(\mathcal{C})
	=&\log \binom{2^L}{M}-\log \binom{2^{L'}-(M-1)P+M-1}{M-1} 
	\\
	 &- [M(L-L')-4KL'-4K\log (4KL')-o(\log (4KL'))-4K\lceil\log ML\rceil-o(\log ML)]\\
	 \le &4K\log ML + (12K+2)\log M+O(K^3)+o(\log ML).
\end{align*}
\begin{remark}
    The decoding procedure can be modified to correct a combination of at most $K$ deletions and insertions. In Step (2), instead of looking for a length $L'-K$ subsequence of $(x_{i,1},\ldots,x_{i,L'})$, we find the unique $\pi(i)$ such that there exists an $\ell\in[L'-K,L'+K]$ satisfying $d_D((x'_{\pi(i),1},\ldots,x'_{\pi(i),\ell}),(x_{i,1},\ldots,x_{i,L'}))\le K$, where $d_D(\boldx,\boldy)$ is the deletion distance between two binary strings $\boldx$ and $\boldy$, defined as the minimum sum of number of deletions in $\boldx$ and $\boldy$, respectively, such that the resulting strings are equal.
    Since $\{\bolda_i\}^M_{i=1}\in\cS^D$, such $\pi(i)$ is unique.
    
    It can be proved that the concatenation $\boldm'$ in \eqref{equation:concatenation}
	 has deletion distance at most $K$ of $Enc(\boldd_2)$.
	Note that a $K$-deletion code corrects a combination of at most $K$ deletions and insertions \cite{L66}, i.e., recovers a codeword from any sequence that is within deletion distance $K$ of the codeword. Then, $\boldd_2$ can be recovered.
\end{remark}

\subsection{Computing~$F^D_S$}\label{subsection:fds}
We now prove Lemma~\ref{lemma:fds}.
The robust indexing algorithm for generating the indexing strings~$\{x_{i,1},\ldots,x_{i,L'}\}$ is the same as in Section~\ref{section:computingfh} except that we replace the notations~$N_H(\bolda,A)$ and~$Q$, which are based on Hamming distance, with their deletion distance counterparts that will be defined later. For a string~$\boldc\in\{0,1\}^\ell$ and a set of indices~$\Delta=\{\delta_1,\ldots,\delta_r\}\subset [\ell]$, let~$\boldc(\Delta)$ be the length~$\ell-r$ subsequence of~$\boldc$ obtained by deleting bits~$(c_{\delta_1},c_{\delta_2},\ldots,c_{\delta_r})$ in~$\boldc$.

For sequences~$\boldc_1\in\{0,1\}^{\ell_1}$ and~$\boldc_2\in\{0,1\}^{\ell_2}$ and nonnegative integers~$r_1,r_2$, define the set
\begin{align*}
	\cI(\boldc_1,\boldc_2,r_1,r_2)=\{(\Delta_1,\Delta_2):\Delta_1\subseteq[\ell_1],|\Delta_1|\le r_1,\Delta_2\subseteq[\ell_2],|\Delta_2|\le r_2,\boldc_1(\Delta_1)=\boldc_2(\Delta_2)\}
\end{align*}
and the number
\begin{align}\label{equation:definen}
	N(\boldc_1,\boldc_2,r_1,r_2)=|\cI(\boldc_1,\boldc_2,r_1,r_2)|,
\end{align}
which is the number of ways to delete no more than~$r_1$ and~$r_2$ bits in~$\boldc_1$ and~$\boldc_2$, respectively, such that the resulting subsequences are identical. 
For a sequence~$\bolda\in\{0,1\}^{\ell}$ of length~$\ell\in[0,L']$ and a set of sequences~$A\subset \{0,1\}^{L'}$, define
\begin{align*}
	N_D(\bolda,A)=\sum_{\boldc\in A}\sum_{\boldc':\boldc'\in\{0,1\}^{L'} \mbox{ and }(c'_1,\ldots,c'_{\ell})=\bolda}N(\boldc',\boldc,K,K).
\end{align*}
For an empty sequence~$\bolda$ and a sequence~$\boldc$, we have that
\begin{align}\label{equation:emptyND}
	N_D(\bolda,\boldc)=\sum^{K}_{r=0}\binom{L'}{r}^22^{r}\triangleq P,
\end{align}
since~$N_D(\bolda,\boldc)$ is equal to the number of tuples $(\boldc',\Delta_1,\Delta_2)$ of sequences $\boldc'\in\{0,1\}^{L'}$ and index sets $\Delta_1,\Delta_2\subset [L']$ such that after no more than~$K$ deletions in indices $\Delta_1$ and $\Delta_2$ in $\boldc$ and $\boldc'$, respectively, we obtain the same subsequence $\boldc(\Delta_1)=\boldc'(\Delta_2)$. As mentioned above, $P$ serves as the counterpart of $Q$ for the deletion channel.

The algorithm for computing~$F^D_S$ follows a similar outline to that for computing~$F^H_S$. We first generate a set of numbers $q_1,\ldots,q_M$ such that $q_i\ge q_{i+1}+P$ for $i\in[M-1]$. Then, we generate strings $\bolda_1,\ldots,\bolda_M$ from $\{q_1,\ldots,q_M\}$,  by using the encoding procedure in Section \ref{section:computingfh} and replacing the numbers~$N_H((a_{i,1},\ldots,a_{i,\ell-1},0),\{\bolda_j\}^{i-1}_{j=1})$ and~$Q$ with~numbers $N_D((a_{i,1},\ldots,a_{i,\ell-1},0),\{\bolda_j\}^{i-1}_{j=1})$ and~$P$, respectively. To prove the correctness of the algorithm, 
we need to show that~$N_D(\bolda,A)$ satisfies the two properties similar to the ones in Eq.~\eqref{equation:recursion} and~Eq.~\eqref{equation:computeN}.
The first is that
\begin{align}\label{equation:sumnd}
N_D(\bolda,A)=N_D((\bolda,0),A)+N_D((\bolda,1),A)
\end{align}
for a sequence~$\bolda\in\{0,1\}^\ell$ of length~$\ell\in[L'-1]$ and a set~$A\subset\{0,1\}^{L'}$, which is a deletion counterpart of Eq.~\eqref{equation:recursion}. This can be proved by noticing that
\begin{align*}
	N_D(\bolda,A)=\sum_{\boldc':\boldc'\in\{0,1\}^{L'} \mbox{ and }(c'_1,\ldots,c'_{\ell})=\bolda}\sum_{\boldc\in A}N(\boldc',\boldc,K,K)
\end{align*}
and that for every sequence~$\boldc'\in\{0,1\}^{L'}$ that satisfies~$(c'_1,\ldots,c'_{\ell})=\bolda$, we have either~$c'_{\ell+1}=1$ or~$c'_{\ell+1}=0$. 

The second property is that the number~$N_D(\bolda,A)$ is computable in polynomial time. Since obtaining an explicit expression as in Eq.~\eqref{equation:computeN} is challenging, we compute the number~$N_D(\bolda,\boldc)$ using dynamic programming for two sequences~$\bolda\in\{0,1\}^\ell$ and~$\boldc\in\{0,1\}^{L'}$ such that~$\ell\in[0,L']$. Given~$\bolda$ and~$\boldc$, we compute
\begin{align}\label{equation:definenk1k2r1r2}
	n(\bolda,\boldc,\bolda,\boldc,k_1,k_2,r_1,r_2)=\sum_{\boldc':\boldc'\in\{0,1\}^{L'-\ell+k_1} \mbox{ and }(c'_1,\ldots,c'_{k_1})=(a_{\ell-k_1+1},\ldots,a_{\ell})}N(\boldc',(c_{L'-k_2+1},\ldots,c_{L'}),r_1,r_2)
\end{align}
Note that~$N_D(\bolda,\boldc)=n(\bolda,\boldc,\ell,L',K,K)$. In addition, by definition of $N_D(\bolda,A)$,  we have that $N_D(\bolda,A)=\sum_{\boldc\in A}N_D(\bolda,\boldc)$. Hence, we wish to compute $N_D(\bolda,\boldc)$ efficiently. Efficient computation of $N_D(\bolda,A)$ follows whenever $|A|=M$ is of polynomial size.

For~$k_1=0$, we have that
\begin{align}\label{equation:k1equals01}
n(\bolda,\boldc,0,k_2,r_1,r_2)=\sum_{\boldc':\boldc'\in\{0,1\}^{L'-\ell}}N(\boldc',(c_{L'-k_2+1},\ldots,c_{L'}),r_1,r_2),
\end{align}
which by Eq.~\eqref{equation:definen} and the definition of $\cI(\boldc_1,\boldc_2,r_1,r_2)$ equals~$0$ when~$L'-\ell-r_1>k_2$ or~$k_2-r_2>L'-\ell$. When~$L'-\ell-r_1\le k_2$ and~$k_2-r_2\le L'-\ell$, we show that

\begin{align}\label{equation:k1equals02}
n(\bolda,\boldc,0,k_2,r_1,r_2)&=\sum^{\min\{r_2,k_2-(L'-\ell-r_1)\}}_{i=k_2-(L'-\ell)}\binom{k_2}{i}\binom{L'-\ell}{L'-\ell-(k_2-i)}2^{L'-\ell-(k_2-i)},\mbox{ if~$k_2\ge L'-\ell$;  and}\\
n(\bolda,\boldc,0,k_2,r_1,r_2)&=\sum^{r_1}_{i=L'-\ell-k_2}\binom{k_2}{k_2-(L'-\ell-i)}\binom{L'-\ell}{i}2^{i}, \mbox{ otherwise.}\label{equation:k1equals03}
\end{align}


For~$k_2\ge L'-\ell$ and sets~$(\Delta_1,\Delta_2)\in\cI(\boldc',(c_{L'-k_2+1},\ldots,c_{L'}),r_1,r_2)$, the cardinality~$|\Delta_2|$ satisfies $$k_2-(L'-\ell)\le |\Delta_2|\le \min\{r_2,k_2-(L'-\ell-r_1)\},$$ because $\boldc'(\Delta_1)=(c_{L'-k_2+1},\ldots,c_{L'})(\Delta_2)$ (Recall that $\boldc(\Delta)$ is the subsequence of $\boldc\in\{0,1\}^\ell$ obtained after deleting bits with indices $\Delta\subset[\ell]$ in $\boldc$). For given~$|\Delta_2|\in [k_2-(L'-\ell),\min\{r_2,k_2-(L'-\ell-r_1)\}]$, there are~$\binom{k_2}{|\Delta_2|}$ ways to select~$\Delta_2$ and~$\binom{L'-\ell}{L'-\ell-(k_2-|\Delta_2|)}$ choices of~$\Delta_1$ such that there exists a  $\boldc'\in\{0,1\}^{L'-\ell}$ satisfying
$$(c_{L'-k_2+1},\ldots,c_{L'})(\Delta_2)=\boldc'(\Delta_1).$$ Moreover, given~$(c_{L'-k_2+1},\ldots,c_{L'}),\Delta_1$, and~$\Delta_2$, there are~$2^{L'-\ell-(k_2-|\Delta_2|)}$ choices of~$\boldc'$ such that $$(c_{L'-k_2+1},\ldots,c_{L'})(\Delta_2)=\boldc'(\Delta_1),$$ and hence Eq.~\eqref{equation:k1equals02} follows.

Similarly, when $k_2<L'-\ell$, the cardinality $|\Delta_1|$ satisfies $$L'-\ell-k_2\le |\Delta_1|\le \min\{r_1,L'-\ell-(k_2-r_2)\}.$$ For each $|\Delta_1|\in[L'-\ell-k_2,\min\{r_1,L'-\ell-(k_2-r_2)\}]$ and for $(c_{L'-k_2+1},\ldots,c_{L'})$, there are $\binom{L'-\ell}{|\Delta_1|}$ choices of $\Delta_1$ and $\binom{k_2}{k_2-(L'-\ell-|Delta_1|)}$ choices of $\Delta_2$, and $2^{|\Delta_1|}$ choices of $\boldc'\in\{0,1\}^{L'-\ell}$ satisfying $$(c_{L'-k_2+1},\ldots,c_{L'})(\Delta_2)=\boldc'(\Delta_1),$$ and hence Eq.~\eqref{equation:k1equals03} follows. 
Therefore, the number~$n_(k_1,k_2,r_1,r_2)$ can be computed when~$k_1=0$.

For~$k_1>0$,
we compute~$n(\bolda,\boldc,k_1,k_2,r_1,r_2)$ iteratively from~$k_1=0$ to~$k_1=\ell$ using the following recursion.
\begin{align}\label{equation:recursionn}
n(\bolda,\boldc,k_1,k_2,r_1,r_2)=&\sum_{k:k\in [L'-k_2+1,L'], c_{k}=a_{\ell-k_1+1}}n(\bolda,\boldc,k_1-1,L'-k,r_1,r_2-k+L'-k_2+1)\nonumber\\
&+2n(\bolda,\boldc,k_1-1,k_2,r_1-1,r_2),
\end{align} 
where $n(\bolda,\boldc,k_1,k_2,r_1,r_2)=0$ if $r_1<0$ or $r_2<0$. We now show that \eqref{equation:recursionn} holds. Recall the definition of $n(\bolda,\boldc,k_1,k_2,r_1,r_2)$ in \eqref{equation:definenk1k2r1r2}. 
Note that for any~$(\Delta_1,\Delta_2)\in \cI(\boldc',(c_{L'-k_2+1},$ $\ldots,c_{L'}),r_1,r_2)$, we have either~$1\in\Delta_1$ or~$1\notin\Delta_1$. When~$1\in\Delta_1$, then$$(c'_2,\ldots,c'
_{L'-\ell+k_1})(\Delta_1\backslash\{1\}-1)=(c_{L'-k_2+1},\ldots,c_{L'})(\Delta_2),$$ where 
$\Delta-i=\{j-i:j\in \Delta\}$ for any set $\Delta$ and integer $i$. Note that there are~$n(\bolda,\boldc,k_1-1,k_2,r_1-1,r_2)$ choices of~$((c'_2,\ldots,c'
_{L'-\ell+k_1}),\Delta_1\backslash\{1\}-1,\Delta_2)$ such that 
\begin{align*}
    (c'_2,\ldots,c'
_{k_1})&=(a_{\ell-k_1+2,\ldots,a_\ell}),\mbox{ and}\\
(c'_2,\ldots,c'
_{L'-\ell+k_1})(\Delta_1\backslash\{1\}-1)&=(c_{L'-k_2+1},\ldots,c_{L'})(\Delta_2).
\end{align*}

Since $c'_1$ can be either $0$ or $1$, therefore
\begin{itemize}
    \item When $1\in\Delta_1$, we have $2n(\bolda,\boldc,k_1-1,k_2,r_1-1,r_2)$ choices of $(\boldc',\Delta_1,\Delta_2)$ such that $(c'_2,\ldots,c'
_{k_1})=(a_{\ell-k_1+2,\ldots,a_\ell})$ and $\boldc'(\Delta_1)=(c_{L'-k_2+1},\ldots,c_{L'})(\Delta_2)$. 
\item When~$1\notin\Delta_1$, let $k$ be the minimum index such that $k\in [L'-k_2+1,L']$ and $(k-L'+k_2)\notin\Delta_2$. Then, we have that $c_k=c'_1=a_{l-k_1+1}$,  $[k-L'+k_2-1]\in(\Delta_2\cup\{0\})$, and $(c'_2,\ldots,c'
_{L'-\ell+k_1})(\Delta_1-1)=(c_{k+1},\ldots,c_{L'})(\Delta_2\backslash[k-L'+k_2-1]-k+L'-k_2)$. There are $n(\bolda,\boldc,k_1-1,L'-k,r_1,r_2-k+L'-k_2+1)$ choices of $((c'_2,\ldots,c'
_{k_1}),\Delta_1-1,\Delta_2\backslash[k-L'+k_2-1]-k+L'-k_2)$ such that $(c'_2,\ldots,c'
_{k_1})(\Delta_1-1)=(c_{k+1},\ldots,c_{L'})(\Delta_2\backslash[k-L'+k_2-1]-k+L'-k_2)$ and that $(c'_2,\ldots,c'
_{k_1})=(a_{\ell-k_1+2,\ldots,a_\ell})$. Therefore, there are $n(\bolda,\boldc,k_1-1,L'-k,r_1,r_2-k+L'-k_2+1)$ choices of $(\boldc',\Delta_1,\Delta_2)$ such that $(c'_1,\ldots,c'_{k_1})=(a_{\ell-k_1+1},\ldots,a_\ell)$ and $\boldc'(\Delta_1)=(c_{L'-k_2+1},\ldots,c_{L'})(\Delta_2)$. Note that for each $k$ satisfying $k\in[L'-k_2+1,L']$ and $c_k=c'_1=a_{\ell-k_1+1}$, there are $n(\bolda,\boldc,k_1-1,L'-k,r_1,r_2-k+L'-k_2+1)$ choices of such $(\boldc',\Delta_1,\Delta_2)$. In addition, different $k$ corresponds to different choices since $k$ is the minimum index such that $(k-L'+k_2)\notin\Delta_2$. Hence, there are $\sum_{k:k\in [L'-k_2+1,L'], c_{k}=a_{\ell-k_1+1}}n(\bolda,\boldc,k_1-1,L'-k,r_1,r_2-k+L'-k_2+1)$ choices of $(\boldc',\Delta_1,\Delta_2)$ such that $(c'_1,\ldots,c'
_{k_1})=(a_{\ell-k_1+1,\ldots,a_\ell})$ and $\boldc'(\Delta_1)=(c_{L'-k_2+1},\ldots,c_{L'})(\Delta_2)$.
\end{itemize}
 Summing up the number of choices of $(\boldc',\Delta_1,\Delta_2)$ over the two cases, we have \eqref{equation:recursionn}.

By \eqref{equation:k1equals02},~\eqref{equation:k1equals03}, and \eqref{equation:recursionn}, the number~$N_D(\bolda,\boldc)=n(\bolda,\boldc,\ell,L',K,K)$ can be recursively computed for any~$\bolda\in\{0,1\}^\ell$ and~$\boldc\in\{0,1\}^{L'}$. Therefore, the encoding/decoding can be computed in~$poly(M,L')$ time. 

We are now ready to present the algorithm that computes $F^D_S(d)$ for an integer $d\in \left[\binom{2^{L'}-MP+M-1}{M-1}\right]$. The algorithm is similar to the encoding procedure in Section \ref{section:computingfh}, by replacing  $N_H(\bolda,A)$ with $N_D(\bolda,A)$ for any sequence $\bolda$ and set of sequences $A$. In addition, the integers $q_i$ are generated such that $q_1=2^{L'}$ and $q_{i}-q_{i+1}\ge P$ for $i\in [M-1]$. Such $q_i$, $i\in[M]$ can be generated following the same argument in Lemma \ref{lemma:dtoq}, since $d\in \left[\binom{2^{L'}-MP+M-1}{M-1}\right]$. Given integers $q_i$, $i\in[M]$,  satisfying $q_1=2^{L'}$ and $q_{i}-q_{i+1}\ge P$ for $i\in [M-1]$, the encoding procedure for generating $\{\bolda_1,\ldots,\bolda_M\}$ is given in Algorithm \ref{alg:ptoadeletion}.
\begin{algorithm}
  \caption{Encoding from $\{q_1,\ldots,q_M\}$ to $\{\bolda_1,\ldots,\bolda_M\}$}\label{alg:ptoadeletion}
  \begin{algorithmic}[1]
      \For{$i\in [M]$}
        \State $q=q_i$
        \For{$\ell\in[L']$}
            \If{$2^{L'-\ell}-N_D((a_{i,1},\ldots,a_{i,\ell-1},0),\{\bolda_j\}^{i-1}_{j=1})\ge q$}
            \State $a_{i,\ell}=0$
            \Else
            \State $q=q-(2^{L'-\ell}-N_D((a_{i,1},\ldots,a_{i,\ell-1},0),\{\bolda_j\}^{i-1}_{j=1}))$
            \State $a_{i,\ell}=1$
            \EndIf
        \EndFor
      \EndFor
      \State \textbf{return} $\{\bolda_1,\ldots,\bolda_M\}$
  \end{algorithmic}
\end{algorithm}
The correctness of Algorithm \ref{alg:ptoadeletion} follows similar arguments to the proof of correctness of Algorithm \ref{alg:ptoasubstitution}. To this end, we prove that the input $\{q_1,\ldots,q_M\}$ and the output $\{\bolda_1,\ldots,\bolda_M\}$  satisfy
\begin{align}\label{equation:aideletion}
\decimal(\bolda_i)=q_i-1+\sum_{\ell:a_{i,\ell}=1\mbox{ and }\ell\in[L']} N_D((a_{i,1},\ldots,a_{i,\ell-1},0),\{\bolda_j\}^{i-1}_{j=1})
\end{align}
and 
$\{\bolda_1,\ldots,\bolda_M\}\in \cS^D$.  The following is a deletion distance counterpart of Lemma \ref{lemma:qequals1}, by replacing $N_H(\bolda,A)$ with $N_D(\bolda,A)$ for any sequence $\bolda\in\{0,1\}^\ell$ and set $A\in \{0,1\}^{L'}$.
\begin{lemma}\label{lemma:qequals1deletion}
For $\ell\in[L']$ and $i\in[M]$, 	
 after the~$\ell$-th inner for loop in the~$i$-th outer for loop in Algorithm \ref{alg:ptoadeletion}, we have that
	\begin{align}\label{equation:qdynamicsdeletion}
		0<q\le 2^{L'-\ell}-N_D((a_{i,1},\ldots,a_{i,\ell}),\{\bolda_j\}^{i-1}_{j=1})
	\end{align} 
	At the end of the~$i$-th outer for loop, we have that~$q=1$.
\end{lemma}
\begin{proof}
The proof is similar to that of Lemma \ref{lemma:qequals1}, by noticing that 
	\begin{align*}
		N_D(0,\{\bolda_j\}^{i-1}_{j=1})+N_D(1,\{\bolda_j\}^{i-1}_{j=1})=&\sum^{i-1}_{j=1}N_D(\emptyset,\bolda_j)\\
		\overset{(a)}{=}&(i-1)P
	\end{align*}
	where $\emptyset$ is the empty string and $(a)$ follows from \eqref{equation:emptyND} and the fact that $N_D(\bolda,A)=\sum_{\boldc\in A}N_D(\bolda,\boldc)$. In addition, we have \eqref{equation:sumnd}, which is the deletion metric version of \eqref{equation:recursion}. The rest of the proof follows the same as in Lemma \ref{lemma:qequals1}.
\end{proof}
From Lemma \ref{lemma:qequals1deletion}, we have
\begin{align*}
q=2^{L'-L'}-N_D(\bolda_i,\{\bolda_j\}^{i-1}_{j=1})=1,
\end{align*}
at the end of the $i$-th outer for-loop, $i\in[M]$.
Hence, $N_D(\bolda_i,\{\bolda_j\}^{i-1}_{j=1})=0$ for $i\in[M]$ and $\cD_K(\bolda_i)\cap\cD_K(\bolda_j)=\emptyset$ for any $i\ne j$, $i,j\in[M]$.
Then, we have that~$\{\bolda_i\}^M_{i=1}\in\cS^D$. In addition, 
similar to Lemma \ref{lemma:ai}, we can use Lemma \ref{lemma:qequals1deletion} to show that the output $\{\bolda_i\}^M_{i=1}$ satisfies Eq. \eqref{equation:aideletion}.

Therefore, we have the following decoding algorithm, similar to the one in Section \ref{section:computingfh}.

\textbf{Decoding:} 
\begin{itemize}
	\item [\textbf{(1)}] Order the strings~$\{\bolda_{i}\}^{M}_{i=1}$ such that~$\bolda_1>\bolda_2>\ldots>\bolda_M$.
	\item [\textbf{(2)}]
	For~$i\in[M]$,
	\begin{align}\label{equation:decodingqi}
	q_i=\decimal(\bolda_i)+1+\sum_{\ell:a_{i,\ell}=1\mbox{ and }\ell\in[L']}N_D((a_{i,1},\ldots,a_{i,\ell-1},0),\{\bolda_j\}^{i-1}_{j=1}).
	\end{align}
\end{itemize}
Finally, the correctness of decoding is guaranteed by \eqref{equation:aideletion} and the fact that $\bolda_1>\bolda_2>\ldots>\bolda_M$, where $\bolda_i$ is the output generated in the $i$-th outer-loop. The latter follows similar proof to the one in Section \ref{section:computingfh}. Suppose there exists $i_1>i_2$ such that $\bolda_{i_1}>\bolda_{i_2}$. Then we have that  
\begin{align}\label{equation:qi2qi1deletion}
	q_{i_2}-q_{i_1}<& \sum_{\ell:a_{i_1,\ell}=1\mbox{ and }\ell\in[\ell^*]}(2^{L'-\ell}- N_D((a_{i_1,1},\ldots,a_{i_1,\ell-1},0),\{\bolda_j\}^{i_2-1}_{j=1}))\nonumber\\
	&-\sum_{\ell:a_{i_1,\ell}=1\mbox{ and }\ell\in[\ell^*]}(2^{L'-\ell}- N_D((a_{i_1,1},\ldots,a_{i_1,\ell-1},0),\{\bolda_j\}^{i_1-1}_{j=1}))\nonumber\\
	=&\sum_{\ell:a_{i_1,\ell}=1\mbox{ and }\ell\in[\ell^*]}(N_D((a_{i_1,1},\ldots,a_{i_1,\ell-1},0),\{\bolda_j\}^{i_1-1}_{j=1})-N_D((a_{i_1,1},\ldots,a_{i_1,\ell-1},0),\{\bolda_j\}^{i_2-1}_{j=1}))\nonumber\\
	=&\sum_{\ell:a_{i_1,\ell}=1\mbox{ and }\ell\in[\ell^*]}N_D((a_{i_1,1},\ldots,a_{i_1,\ell-1},0),\{\bolda_j\}^{i_1-1}_{j=i_2})\nonumber\\
	\overset{(a)}{\le}&N_D(\emptyset,\{\bolda_j\}^{i_1-1}_{j=i_2})\nonumber\\
	\overset{(b)}{\le}&(i_1-i_2)P,
\end{align}
where $\emptyset$ is the empty sequence and $(a)$ follows from the definition of $N_D(\bolda,A)$ and the fact that the strings which have $(a_{i_1,1},\ldots,a_{i_1,\ell_1-1},0)$ and~$(a_{i_1,1},\ldots,a_{i_1,\ell_2-1},0)$ as prefixes, respectively, are different. Inequality $(b)$ follows from \eqref{equation:emptyND} and the fact that $N_D(\bolda,A)=\sum_{\boldc\in A}N_D(\bolda,\boldc)$.
 
 Similar to the procedure for computing $(F^H_S)^{-1}$ in Section \ref{section:computingfh}, we can compute $(F^D_S)^{-1}$ by using the decoding procedure above and obtain the set of integers $\{q_1,\ldots,q_M\}$, and then recover $d$ from $\{q_1,\ldots,q_M\}$ by following similar steps in Lemma \ref{lemma:dtoq}.
\section{Conclusions and Future Work}\label{section:FutureWork}
	Motivated by DNA storage applications, this paper studied coding for channels where data are encoded as a set of~$M$ unordered strings of length~$L$. A~$K$-substitution correcting code and a $K$-deletion correcting code were presented for this channel. Our codes achieve~$O(K\log ML)$ redundancy for constant~$K$, which are order-wise optimal. Our $K$-deletion correcting code can be slightly modified to correct a combination of at most $K$ deletions and insertions. It is interesting to find optimal codes that correct substitution or deletion/insertion errors for larger range of parameters~$K,M$, and~$L$.


\bibliographystyle{IEEEtran}

\appendices


\section{Proof of~Eq. \eqref{equation:inequality2}}\label{section:proofofinequality}
\begin{align*}
    r(\mathcal{C})
    \le&\log \binom{2^L}{M}-\log \lceil\frac{\prod^{M-1}_{i=1}(2^{L'}-iQ)}{(M-1)!}\rceil\\
    &- [M(L-L')-4KL'-2K\lceil\log ML\rceil]\\
    \le &\log \frac{2^{LM}}{M!} - \log \frac{(2^{L'}-MQ)^{M-1}}{(M-1)!}\\
    &- [M(L-L')-4KL'-2K(\log ML+1)]\\
    = &ML'- \log (2^{L'}-MQ)^{M-1} +4KL'\\
    &+2K\log ML + 2K-\log M\\
    =& \log \frac{2^{L'(M-1)}}{(2^{L'}-MQ)^{M-1}} + L'+4KL'\\
    &+2K\log ML + 2K-\log M\\
    =& (M-1)\log (1+\frac{MQ}{2^{L'}-MQ})+ L'+4KL'\\
    &+2K\log ML + 2K-\log M\\
    \overset{(a)}{\le} &(M-1)\log (1+\frac{1}{M})+ L'+4KL'\\
    &+2K\log ML + 2K-\log M\\
    \le & \log e + L'+4KL'+2K\log ML + 1+2K-\log M\\
    =&2K\log ML + (12K+2)\log M+O(K^3)+O(K\log\log ML)
\end{align*}
where~$(a)$ is equivalent to
\begin{align*}
    MQ(M+1)\le 2^{L'},
\end{align*}
which can be 
obtained from the following inequality
\begin{align}\label{equation:inequality22}
    M^2(3\log M +4K^2+2)^{2K}\le 2^{3\log M+4K^2+1}.
\end{align}
Eq.~\eqref{equation:inequality22} is proved as follows. 
Rewrite Eq.~\eqref{equation:inequality22} as
\begin{align}\label{equation:rewrite}
    (3\log M +4K^2+2)^{2K}\le 2^{\log M+4K^2+1}.
\end{align}
Define functions~$g(y,K)=\ln (3y+4K^2+2)^{2K}$ and~$h(y,K)=\ln 2^{y+4K^2+1}$. Then we have that
\begin{align*}
    \partial h(y,K)/\partial y
    -\partial g(y,K)/\partial y
    =\ln 2 - 6K/(3y+4K^2+2),
\end{align*}
which is positive for~$y\ge 1$ and~$K\ge 2$. Therefore, for~$K\ge 2$ and~$y\ge 1$, we have that
\begin{align*}
    h(y,K)-g(y,K)\ge h(1,K)-g(1,K).
\end{align*}
Furthermore,
\begin{align*}
    \partial h(1,K)/\partial K
    -\partial g(1,K)/\partial K
    =&(8\ln 2) K -2\ln (4K^2+5)-   16K^2/(4K^2+5)\\
    >&(8\ln 2) K-2\ln (5K^2)-4\\
    =&4(K-1-\ln K)+(8\ln 2-4) K-2\ln 5\\
    \overset{(a)}{\ge} &(8\ln 2-4) K-2\ln 5,
\end{align*}
where~$(a)$ follows since~$K=e^{\ln K}\ge 1+\ln K$. Since~$(8\ln 2-4) K-2\ln 5$ is positive for~$K\ge 3$, we have that~$h(1,K)/\partial K
    >\partial g(1,K)/\partial K$ for~$K\ge 3$.
It then follows that~$h(1,K)-g(1,K)\ge\min\{h(1,2)-g(1,2),h(1,3)-g(1,3)\}>0$ for~$K\ge 2$. Hence~$h(y,K)>g(y,K)$ for~$y\ge 1$ and~$K\ge 2$, which implies that Eq.~\eqref{equation:rewrite} holds when~$M\ge 2$ and~$K\ge 2$. 


When~$K=1$ we have that
\begin{align*}
    2^{\log M+4K^2+1}
    =&32(1+\sum^\infty_{i=1}\log^i M/i!)\\
    \ge& 32(1+\log M+\log^2M/2)\\
    \ge &(4\log M+5)^2\\
    =&(3\log M+4K^2+2)^{2K}.
\end{align*}
Hence, Eq.~\eqref{equation:rewrite} and Eq.~\eqref{equation:inequality22} holds. We now finish the proof of Eq.~\eqref{equation:inequality2}.

\end{document}